\documentclass[11pt]{article}
\usepackage{graphicx} 
\usepackage{titlesec}

\usepackage{graphicx}
\usepackage{subcaption}

\usepackage{amsmath}
\usepackage{amssymb}
\usepackage{amsthm}
\usepackage{anysize}

\usepackage{bbm}

\usepackage{natbib}
\usepackage{pdflscape}

\usepackage{graphicx}
\usepackage{lscape}

\usepackage{multicol}
\usepackage{multirow}
\usepackage{scalefnt}
\usepackage{color}
\usepackage[dvips]{epsfig}
\usepackage{comment}

\usepackage{amsmath}
\usepackage{chngcntr}  

\theoremstyle{plain}
\newtheorem{thm}{Theorem}[section]
\newtheorem{cor}[thm]{Corollary}
\newtheorem{lem}[thm]{Lemma}
\newtheorem{prop}[thm]{Proposition}
\theoremstyle{definition}
\newtheorem{defn}[thm]{Definition}
\theoremstyle{definition}

\theoremstyle{remark}

\marginsize{30mm}{30mm}{30mm}{30mm}

\newcommand{\bp}{\mathbb P}

\newcommand{\R}{\mathbb R}
\newcommand{\E}{\mathbb E}
\newcommand{\F}{\mathcal F}

\newcommand{\XIM}{X^{\scriptscriptstyle \mathcal{I}}}
\newcommand{\XA}{X^{\scriptscriptstyle \mathcal{A}}}
\newcommand{\SIM}{S^{\scriptscriptstyle \mathcal{I}}}
\newcommand{\SA}{S^{\scriptscriptstyle \mathcal{A}}}

\title{
The skew Brownian motion should not be used as a risk-neutral returns process: a well-posed skew-normal alternative\footnote{This paper was presented at the 50th AMASES Conference, Rome, Italy. The authors would like to thank Paolo Pigato and Damiano Rossello for the useful comments, and particularly Sara Mazzonetto for valuable discussions and pointing out to important literature.}\\
}

\author{Michele Bufalo\footnote{\texttt{email: michele.bufalo@uniba.it}} \: \: \: Lorenzo Torricelli\footnote{\texttt{email: lorenzo.torricelli2@unibo.it}}}
\date{}

\begin{document}

\maketitle

\begin{abstract}
 Return models for risk-neutral financial valuation based on skew Brownian motions (SBMs) have been introduced about twenty years ago, and have recently enjoying growing popularity. Unfortunately, the story behind their development is one of mistakes and erroneous interpretations, beginning from the foundational misrepresentations that the prevalent financial model is based on the  It\^o-McKean SBM -- which, in fact, it is not. Besides, and more seriously, the clarification of \cite{rossello2012arbitrage} that price models with a local time in their returns, such as the SBM, are arbitrageable 
 has been, by and large, ignored.  
 
 In this paper, we try to clear the field from the confusions and misconceptions lurking in the standing option pricing literature on SBM,  by exposing all the errors we could trace in the treatment so far.
 Recognizing however the potential of the SBM skew-normal marginals for risk-neutral valuation, as a positive contribution, we  reformulate the putative SBM call pricing formula and show that, even if the SBM return model admits arbitrage,  its option pricing formula does not. The correct  Markovian SDE with skew-normal marginals is then identified, its strong well-posedness shown, and by exploiting the availability of closed formulae, an asymptotic analysis of the implied volatility surface is offered. 
 En route to our conclusions, we obtain a novel normal/skew-normal stochastic dominance property of independent interest. 

  \end{abstract}


\section{Introduction}

 Of the many limitations of the Black-Scholes option pricing model, returns symmetry is arguably at the top of the list. All of the now classic alternative models (e.g. exponential L\'evy models, stochastic volatility, jump diffusion)  address this issue by radically changing the Brownian noise underlying the returns. In principle it does however also makes sense to use a driver more similar to a Brownian motion, but whose probabilities of positive or negative excursions away from zero are not symmetric, but depend on a skew parameter.  The best-known process of such kind is certainly the \cite{ito1963brownian} skew Brownian motion (SBM). Although the full mathematical definition of an SBM is intrinsically operatorial -- or so it was  before \cite{harrison1981skew}  characterized it as a local time equation solution-- its heuristics are pretty clear and attractive for financial modeling. An SBM is a stochastic process that diffuses as a reflected Brownian, or its negative, on its excursions away from zero and at the time of the next visit of the origin, the sign of the next excursion is determined by an independent Bernoulli trial. 
 
 The story of the use of skew processes in finance is certainly a troubled one. As far as these authors know, it starts in \cite{corns2007skew}, and, with the hindsight of twenty years of research, looks puzzling right from the beginning. Despite what one might think, and even in contrast to what the authors themselves claim, the model proposed there is not based on the classic Ito-McKean skew Brownian motion (henceforth IMSBM). The actual chosen driver is rather a weighted sum of a Brownian motion and an independent reflecting Brownian motion; the law of this process is well-known to be skew-normal,  as discussed prominently in the monograph \cite{azzalini2013skew}, so that the name Azzalini skew-Brownian motion (ASBM)  seems appropriate. In contrast, that of an IMSBM is a mixture of half-normal distributions.  Although some connections in terms of static mixing can be established between the two (as discussed in \cite{corns2007skew}), the fact remains that we are facing two genuinely distinct processes, not just two different representations of the same process. This is the first --but certainly not the last-- conceptual  misunderstanding surrounding the SBM literature in finance.      

  Five years later, \cite{rossello2012arbitrage} showed that option pricing models based on a geometric  SBM are arbitrageable, in the  sense of the absence of an equivalent martingale measure (EMM) for the process.  The author explicitly addresses the IMSBM, although it is fair to say that  his argument applies equally well to show arbitrage in the ASBM model. A known necessary condition for the existence of an E(L)MM for a semimartingale model is the so-called ``structure condition'', that is, absolute continuity of the finite variation part with respect to the quadratic variation of the local martingale part in the Doob decomposition. 
  This property was known in various versions at least since \cite{schweizer1992martingale} and \cite{delbaen1995existence}. The heart of Rossello's argument is the  remark that if in the 
  Doob's decomposition appear both a local time and a Brownian motion, the structure condition fails. As it happens, such is the case for the IMSBM. 

 
 Despite this established result, a large and growing body of option pricing literature  based on the ASBM has emerged, arguably inspired by  \cite{zhu2018new}, an article that has drawn considerable attention. In such work the authors claim to have (a) fixed a mistake in the option pricing formula of \cite{corns2007skew}, (b) shown the existence of  an EMM for the GSBM. As it happens, the error in the formula of \cite{corns2007skew} was genuine, but unrelated and independent from the deeper problem of the violation of the arbitrage property. Therefore,  even though Zhu and He correctly address (a), in light of the line of reasoning of  Rossello, there is nothing that can be done about (b), even though the authors maintain otherwise.  This unfortunate circumstance contributed to the wrong impression that the ASBM model is indeed legitimate. To make the story even more intricate, \cite{zhu2018new} committ a new mistake in the option pricing formula, one that was not present in \cite{corns2007skew}. We will explain this error in Appendix B.

A similar situation  has recently drawn attention, for another arbitrageable model the geometric reflected Brownian motion (RBM),  proposed by \cite{veestraeten2008valuing}. That this model also does not satisfy the structure condition 
has been shown in \cite{buckner2024arbitrage}.   Moreover, the authors observe that the European pricing formula are themselves arbitrageable.
 As we shall see, this is where a line can be drawn between the RBM and the SBM models, in that the pricing formulae based on the ASBM turn out to be viable instead. 
 Similar violations of no-arbitrage for reflected models have been observed by \cite{melnikov2021modifications} when analyzing a reflected version of the Bachelier model.

 That the ASBM option pricing model admits arbitrage is not to say that SBM cannot be \emph{tout court} used in finance. On the contrary, its application is notable and of interest in factor modeling non-traded risk factors such as in interest rates (\cite{tian2018skew}) or in e.g. portfolio theory  (\cite{bufalo2022forecasting}) when statistical returns are all that matters. Also in the field of Markovian option pricing local times and SBMs can naturally appear, when singular diffusions are chosen as risk drivers  of local volatility models. For example \cite{gairat2017density} show that an SBM with a two-valued drift arises as a solution of an SDE where returns are volatility-normalized, and study its modeling performance in \cite{gairat2024extreme}.

 Be that as it may, the geometric ASBM ``option pricing model'' has become increasingly popular in recent times, and we are effectively already at a point where a wrong model crystallized into a wrong research strain. At the time we write, there appear to be at least 40 published papers citing \cite{zhu2018new} as a reference, in the vast majority of which it appears as a core one. 
 
 The present work is motivated by the recognition that matters regarding SBM models in finance need to be finally clarified. In the first part of this paper we augment the standing critique of the skew-Brownian motion geometric models in finance by  making three important new points. Firstly, we observe that the model that the literature generally discusses is not based on the IMSBM, but rather on the ASBM. That being recognized, we offer a variation of the proof of Rossello (originally devised for IMSBM returns) tailored for the ASBM, the process on which the model of \cite{corns2007skew} is actually based. Secondly, connecting to the arbitrage hierarchy in \cite{fontana2015weak}, we clarify that the profit opportunity in skew-models is not just riskless, but also increasing. Thirdly, we explicitly show one such arbitrage opportunity, which should ultimately lift all doubts about the status of these models. 
 
 In the second part, as a positive contribution, we instead seek some ways to salvage the vanilla option valuation formula \emph{inspired} by the ASBM marginals of skew-normal type. After, all one of the benefits of using a skew Brownian motion (if that was possible) for modeling asset returns  would be that its marginal laws --hence option prices-- are explicitly known, which, in turn, makes it possible to unlock the analytic structure of the all-important implied volatility surface (see e.g. \cite{rop:10}).

 In fact, it turns out that this can be done, as it is shown by appealing to the now-classic theory of no-arbitrage  characterization in terms of the vanilla payoffs. Works such as \cite{carr2005note}, \cite{dav+hob:07}, \cite{lowther2008fitting}, and \cite{roper2009relationship}, suggest that it is possible to directly connect the pricing function to the \cite{kellerer1972markov} Theorem, which establishes that a time family of laws increasing in convex order has a supporting strong Markovian martingale. The  papers above provide various equivalent checklists for a theoretical call function in order to be free of arbitrage.  By satisfying all the items on one such list, we prove that the skew-normal option pricing formula admits a supporting Markovian martingale process, and therefore that it has no arbitrage. 

 In light of this fact, we proceed to establish that the local volatility SDE with skew-normal marginals (the skew-normal local volatility model, SNLV) is (weakly) well-posed and has marginals matching the ASBM option pricing model. Our main finding is that  the skew-normal local volatility surface is singular at the origin. If, on the one hand, this poses no problem (in our setting) for strong existence and uniqueness, on the other, it  links to literature on the relationship between singular local volatility and the phenomenon of the implied volatility skew explosion, that is, the observed fact that at short maturity the slope of the ATM implied volatility becomes very steep. This fact was first observed by \cite{pig:19} and implemented in \cite{friz2020step}, (see also \cite{gairat2024extreme}), who use a step-discontinuous local volatility function on the whole space manage to generate a power law explosive skew of order $-1/2$. More recently, \cite{torricelli2025parametric} have noticed that refined asymptotics of the order $h-1/2$, $h \in (0,1/2)$ can also be obtained once a diverging in zero local volatility surface is considered. At any rate, the SNLV serves as an example that the assumptions behind \cite{pig:19} original intuition can be weakened, and in order to generate a diverging ATM implied volatility slope, a discontinuity at time zero is sufficient, as opposed to one over the whole time domain.
As a further  contribution, on our way to the proof of the no-arbitrage properties we show some (as far as we know) novel  properties of the extended skew-normal family. In particular, we identify a stochastic dominance property between some normal and extended-skew normal classes that could be potentially of interest outside the scope of this research.

In Section \ref{sec:twoSBM} we discuss the two Brownian motion appearing in finance literature and make the point that the one typically used is not Ito-McKean's. In Section \ref{sec:arbSBM} we thoroughly re-discuss why the GSBM is arbitrageable in the sense of the presence of an increasing profit, and  we  also show one such arbitrage strategy. 
Section  \ref{sec:ZhuHeOP}  introduces from scratch the skew normal option pricing formula, and shows its no-arbitrage property. In Section \ref{sec:localvol} we compute the local volatility from the skew-normal marginals and show that the resulting SDE is strongly well-posed and maches the option prices. In Section \ref{sec:implvol}  some first properties of the SNLV implied volatility surface are studied, and in particular the impact of the discontinuity of the SNLV model coefficient on the ATM skew. Conclusions are drawn in Section \ref{sec:conclusions}. In Appendix A some apparently new properties of the skew-normal distributions are shown, while in Appendix B we explicitly address the errors in \cite{zhu2018new}, leading to the incorrect conclusion that the geometric ASBM model is free of arbitrage, and point out the additional error in their option pricing formula.

\section{A tale of two skew-Brownian motions} \label{sec:twoSBM}

The skew Brownian motion was introduced by K. It\^o and H.P. McKean (\cite{ito1963brownian}, Section 17) as part of a larger treatise on certain processes associated to second order differential operators related to that of the Brownian motion. The original definition of  It\^o and McKean is based on the following path-wise construction. Let $B=(B_t)_{t \geq 0}$ a standard Brownian motion and  $\mathcal  B=\{t>0 | \, B_t=0\}$ the random set of its zeroes. We then write $\mathcal B^c= \bigcup_{i=1}^\infty \mathcal B_i$ with $\mathcal B_i$, $i =1,2, \ldots$ being ordered, disjoint open random intervals called  the \emph{excursions} from zero of $B$. Letting $e_i$, $i=1,2, \ldots$ a sequence of i.i.d. Bernoulli random variables (r.v.) with values in $\{-1, 1\}$ and parameter $\alpha \in (0,1)$, the It\^o-McKean Brownian motion $\XIM_t=(\XIM_t)_{t \geq 0}$ is pathwise constructed as the following process 
\begin{equation}\label{eq:oldIMSBM}
    \XIM_t = \left\{ \begin{array}{cc}
         e_i |B_t|,& t \in \mathcal B_i, \\
        0,   &  t \in \mathcal B.
    \end{array} \right.
\end{equation}

The signs of the excursions of  $B$ have been known since P. L\'evy's work to be Bernoulli i.i.d.s woth sign value and parameter $1/2$. Therefore, if $\alpha=1/2$ then $\XIM$ is itself a standard Brownian motion.
Heuristically, the idea behind the It\^o-McKean skew-Brownian motion (IMSBM) is thus that of taking a reflected Brownian motion and stipulating that each time it hits zero,  the motion either reflects back (continuing its trajectory) or crosses the boundary (switching its sign), each with probabilities $\alpha$ and $1-\alpha$.

The theory of the IMSBM was notably expanded  by \cite{walsh1978diffusion}, \cite{harrison1981skew}, and \cite{le2006one}. In particular \cite{harrison1981skew} provide a key representation of the IMSBM as the unique strong solution of the equation 
\begin{equation}\label{eq:LocSBM}
    \XIM_t=W_t+ (2 \alpha-1) L^{\XIM}_t, \qquad X_0=0, \qquad \alpha \in (0,1),
\end{equation}
for some Brownian Motion $W=(W_t)_{t \geq 0}$,
 started at zero, where, for a continuous local martingale $Y$,  the process   $L^Y=(L^Y_t)_{t \geq 0}$ given by
\begin{equation}\label{eq:ltime}
    L^Y_t=\lim_{\epsilon \rightarrow 0}\frac{1}{2 \epsilon}\int_0^t I_{\{|Y_u|<\epsilon\}} d \langle Y \rangle _u
\end{equation}
is the \emph{symmetric local time} of $Y$ at 0 (\cite{revuz2013continuous}, VI, Corollary 1.9). 
 The symbol $\langle \, \cdot \, \rangle$ indicates throughout the semimartingale quadratic variation. The process $\XIM$ has the desired excursion property 
 \begin{equation}\label{eq:skewexc}
     \mathbb P(\XIM_t >0)=\alpha, \qquad t>0.
 \end{equation}
Representation \eqref{eq:IMSBM} is clearly much more amenable for mathematical purposes than the original definition \eqref{eq:oldIMSBM}. The laws  of $\XIM_t$ were  determined in \cite{walsh1978diffusion} to be absolutely continuous with  densities $p_t$ of the following half-normal mixture form 
\begin{equation}\label{eq:walshdens}
p_t(x) = 
\begin{cases}
\alpha\dfrac{2}{\sqrt{2\pi t}}\,
e^{-x^2/2t}
& \text{if } x > 0, \\[14pt]
(1-\alpha)\dfrac{2}{\sqrt{2\pi t}}\,
e^{-x^2/2t}
& \text{if } x < 0.
\end{cases}
\end{equation}
For financial applications the attractiveness of the IMBSM is apparent: the distributional skew of returns is essentially embedded in \eqref{eq:skewexc}. The original thought inspiring \cite{corns2007skew} was that of transferring \eqref{eq:skewexc}--\eqref{eq:walshdens} to the risk-neutral world, i.e. under a pricing measure $\mathbb Q$, in order to obtain skewed option prices as well. Unfortunately, as we shall see, this is not possible.

Before starting the discussion of why this is so, we begin by making an even more basic remark: despite what the authors claim, the process driving the returns (whether physical or risk-neutral) in \cite{corns2007skew}, the one which also caught up in the literature, is \emph{not} an IMSBM. 

On a filtered probability space $(\Omega, \mathbb P, (\mathcal F_t)_{t \geq 0}, \mathcal F_\infty)$ satisfying the usual conditions,  \cite{corns2007skew} propose the following  asset price process $S=(S_t)_{t\geq 0}$ under $\bp$ as the following exponential 
\begin{equation}\label{eq:GSBM}
    \SA_t=S_0\exp\left( \mu t + \sigma \XA_t \right), \qquad  \mu \in \R, \, S_0, \, \sigma ,  t > 0,
\end{equation}
with $\XA=(\XA_t)_{t \geq 0}$ following a weighted normlalized sum of a standard Brownian motion and an independent reflected Brownian motion, i.e. for independent one-dimensional standard Brownian motions $W_1=(W_{1,t})_{t \geq 0}$, $W_2=(W_{2,t})_{t\geq 0 } $, then  
\begin{equation}\label{eq:IMSBM}
\XA_t= \sqrt{1-\delta^2} W_{1,t}+ \delta|W_{2,t}|, \qquad W_{1,0}=W_{2,0}=0, \quad \delta \in (-1,1),
\end{equation}
We hasten to add that the laws of $\XA_t$ are well-known, and were thoroughly investigated in A. Azzalini's work (e.g. \cite{azzalini2005skew}, \cite{azzalini2013skew}), as a particular representation the \emph{skew-normal distribution}, which we will properly introduce in Section \ref{sec:ZhuHeOP}. For now it suffices to say that the probability densities $q_t$ of $\XA_t$ are given by
\begin{equation}\label{eq:AzzSBMdens}
     q_t(x)=\frac{2}{\sqrt{t}}\phi\left(\frac{x}{\sqrt{t}} \right)\Phi\left(\frac{\delta}{\sqrt{1-\delta^2}} \frac{x}{\sqrt{t}}\right), \qquad x \in \mathbb R, \delta \in (-1,1)
\end{equation}
with $\phi$, $\Phi$ respectively the standard Gaussian density and cumulative distribution functions.

In Proposition 2.1 of \cite{corns2007skew}, it is claimed that $\XA$ \emph{``is distributed as an It\^o-McKean skew Brownian motion''} , which is clearly impossible quite simply because the marginals densities \eqref{eq:walshdens} and \eqref{eq:skewexc} are different.  In order to try and understand  how this confusion has possibly arisen, one might want to look into the discussion   in  \cite{corns2007skew} right before Proposition 2.1. There the authors try to connect the concept of skew Brownian motions with a mechanism of  generation of asymmetric random variables from elliptical families, by claiming that  \emph{``In the construction of Itô–McKean skew Brownian motion, no reference is made to how the
probability weighting $\alpha$ is made. Any formulation of $\alpha$ is consistent with It\^o–McKean skew
Brownian motion provided that  $0 \leq \alpha \leq 1$''} implying that $\alpha$ will be soon changed to a function, and then later that \emph{``any form of probability weighting is consistent with It\^o-McKean as long  as $0 \leq  \alpha \leq 1$ and $\alpha=\alpha(v)$ can be formulated as the conditional distribution of''} [a random variable].
The role of the weighting function $\alpha$, in the authors' intentions, is precisely that of making the marginal law \eqref{eq:AzzSBMdens} emerge.

However, in contrast to the claim, by looking at \eqref{eq:oldIMSBM} one sees that It\^o and McKean exactly specify how  the probability weighting $\alpha$ is made. It is made using the constant function $\alpha(v)=\alpha \in (0,1)$,  corresponding to a pure Bernoulli event, i.e. no non-trivial probability weighting at all. This is the one and only formulation consistent with the process introduced in \cite{ito1963brownian}. The authors seem instead to believe that taking ``mixtures''  of the IMSBM laws over their normalizing constant, which from a distributional perspective would make perfect sense, would not alter the process, whereas in fact it does. Put it differently, since the solution of \eqref{eq:LocSBM} is strongly unique --and hence unique in law-- changing $\alpha$ to be anything other than a constant would result in a process that no longer satisfies  \eqref{eq:IMSBM} and thus, by \cite{harrison1981skew} characterization, one that  no longer is an IMSBM.

The mathematical error in the proof of Proposition 2.1 of \cite{corns2007skew} is the following. After applying Tanaka's Lemma it is correctly proven that, almost surely\begin{equation}\label{eq:wrong}\XA_t=\tilde W_t + \delta L^{W_2}_t, \qquad \delta \in (-1,1), \end{equation}where $\tilde W$ is a Brownian motion correlated with $W_2$ and attained used L\'evy's characterization Theorem. Furthermore, it is correctly observed that $L_t^{\XIM}   \stackrel{d}{=}L_t^{W_2}$,   but then  it is concluded that \eqref{eq:wrong} is \emph{``equivalent''}  to that of \eqref{eq:LocSBM} when $\delta=2 \alpha-1$. The expression used lacks precise mathematical meaning, but if the claimed \emph{equivalence} is taken to mean that the two equations have same solution (whatever the probabilistic sense),  this is incorrect. Indeed, equality in law of the summands, which does hold true for the right hand sides of  \eqref{eq:wrong} and \eqref{eq:LocSBM}, is not sufficient for equality in distribution of the sum: their joint law must also be the same. As a matter of fact, with the hindsight provided by the discussion above, the pairs $(\tilde W_t, L^{W_2}_t)$ in \eqref{eq:wrong} and $(W_t, L^{\XIM}_t)$ in \eqref{eq:LocSBM} 
have different distributions.

Having clarified that $\XA$ and $\XIM$ are different processes, the intuition of \cite{corns2007skew} that $\XA$ produces skewed returns is still correct, and a form  of the excursion property \eqref{eq:skewexc} holds for \eqref{eq:IMSBM} as well. Indeed it can be calculated (using the Owen $T$ function, \cite{azzalini2013skew}) that
\begin{equation}\mathbb P(\XA_t >0)=\frac{2}{\sqrt{t}} \int_0^\infty \phi\left( \frac{x}{\sqrt{t}}\right)\Phi\left(\frac{\delta}{\sqrt{1-\delta^2} }\frac{x}{\sqrt{t}} \right)dx =\frac{1}{2}+ \arctan\left(\frac{\delta}{\sqrt{1-\delta^2}} \right) \pi^{-1}.\end{equation}
Therefore returns can be positively or negatively skewed by acting on $\delta \in (-1,1)$, and $\XA$ can then equally well be used (in principle) for financial modeling.  We will refer to $\XA$ as the \emph{Azzalini skew Brownian motion} (ASBM), a name already proposed by \cite{corns2007skew}. 
 A crucial difference between the two models, that does affect arbitrage considerations is that the IMGSBM is Markovian with respect to its own filtration, while the AGSBM is not, as it needs two separate state variables to have its dynamic evolution correctly specified.

We will thus distinguish the financial models treated: the specification  \eqref{eq:GSBM}--\eqref{eq:LocSBM} will be called the \emph{It\^o-McKean geometric skew-Brownian motion} (IMGSBM) model, and  \eqref{eq:GSBM}--\eqref{eq:IMSBM},  the  Azzalini geometric skew-Brownian motion (AGSBM)  model, and will be indicated respectively by $\SIM=(\SIM_t)_{t \geq 0}$ and $\SA=(\SA_t)_{t \geq 0}$.  We collectively refer to the two classes as geometric skew Brownian motion (GSBM) models. Except for the next section the focus will be on Corns and Satchell suggestion $\SA$, which, to be sure, is the specification of much the literature we consulted.  Incidentally, a possible exception is the reference \cite{rossello2012arbitrage}, that is very important for our goals, and which we will need to partly re-discuss in light of the point just made. In any case, clarifying that there are two distinct models involved is on our opinion of great importance to avert possible future ambiguities.


\section{Arbitrage in  geometric skew Brownian motion models}\label{sec:arbSBM}

 The problems of SBMs model in the context of risk neutral pricing theory were noticed already by  \cite{rossello2012arbitrage} not long after \cite{corns2007skew} was published.  Rossello proved the absence of an EMM for  $\SIM$, going through the route of the violation of the necessary structure condition of \cite{schweizer1992martingale} and \cite{delbaen1995existence} as an effect of the presence of a local time in the dynamics of $\XIM$. Since we are here mostly interested to the AGSMB model $\SA$, in the following we first articulate Rossello's remarks to show the violation of the NFLVR arbitrage in such model. Secondly, in light of recent advances in the theory of arbitrage, we show how the logic in \cite{rossello2012arbitrage} is indeed  far more consequential than the author probably thought it was, meaning that even weaker forms of no-arbitrage are violated. 



We recall in the following the \emph{structure condition} for semimartingales. For a pair of one-dimensional semimartingales $A=(A_t)_{t \geq 0}$ $B=(B_t)_{t \geq 0}$ we say that $A$ is \emph{predictably absolutely continuous} with respect to $B$, and write the usual expression $ dA \ll d B$, if there exists a predictable processes $\alpha=(\alpha_t)_{t \geq 0}$ such that $A$ can be expressed as the stochastic integral \begin{equation}\label{eq:structre}A_t=\int_0^t \alpha_u d B_u.\end{equation} Predictability of the integrand is a slightly stronger requirement than path-wise absolute continuity of the measures $d A$ and $d B$.
\begin{defn}\label{def_struct}
Let $S=(S_t)_{t \geq 0}$, be a continuous semimartingale with Doob decomposition  $S=S_0+F+M$, for some random variable $S_0$ and with $F=(F_t)_{t \geq 0}$ a semimartingale of finite variation and $M=(M_t)_{t \geq 0}$  a local martingale. The process $S$ is said to satisfy the \emph{structure condition} if $ d F  \ll  d \langle  M \rangle$.   \end{defn}

This definition first appeared in \cite{delbaen1995existence} and \cite{schweizer1992martingale}, in a slightly stronger form in the second reference. The structure condition  is a  \emph{necessary} condition for an EMM to exist. While exploring the hierarchy of arbitrage conditions,  it was later found out by  \cite{fontana2015weak},  that a violation of the structure condition is equivalent to the presence of an \emph{increasing} arbitrage. 

We recall that on $(\Omega, \mathbb P, (\mathcal F_t)_{t \geq 0},\mathcal \F_\infty)$ an  \emph{admissible trading strategy}  in the semimartingale asset price $S$ with Doob decomposition as in \eqref{def_struct}, is a predictable process  $H=(H_t)_{t \geq 0}$, square-integrable with respect to 
$\langle M \rangle$, absolutely integrable with respect to $F$, and such that the gains process $\int_0^t H_u d S_u$ is   a.s. bounded from below by a negative constant. It is of relevance here  the class of trading strategies that produce a riskless \emph{increasing} gain.

\begin{defn} Let $S=(S_t)_{t \geq 0}$ be a semimartingale on a filtered space $(\Omega, \mathcal F, \mathbb P, (\mathcal F_t)_{t \geq 0})$ representing a financial asset.  We say that an admissible trading strategy $H$ is an \emph{increasing profit} if $G=(G_t)_{t \geq 0}$ with $G_t=\int_0^t H_u d S_u$ is  an almost-surely non decreasing process such that for some $T>0$ it holds that $\mathbb P\left( G_T>0 \right)>0.$
If such $H$ exists then $S$ has the \emph{increasing profit} $H$. Otherwise we say that the \emph{no increasing profit} condition (NIP) holds for $S$. 
\end{defn}
The NIP condition, that is, the absence of an increasing arbitrage opportunity, is among the weakest arbitrage conditions. Besides the extreme financial severity of one such kind of arbitrage, in technical terms the concept of \emph{weakness} relates to the fact that proving NIP amounts to a check on the semimartingale structure of the asset, 
whereas classic no-arbitrage conditions typically involve the study of the full filtration.
It is a key contribution of \cite{fontana2015weak},  that the structure condition is equivalent to NIP. Combining this fact with the contribution of Rossello we are led to the following theorem. 

\begin{thm}\label{thm:noarbthm}
Let $(\Omega,  \mathbb P, (\mathcal F_t)_{t \geq 0}, \mathcal F_\infty)$ a filtered probability space carrying three standard Brownian motions $W, W_1, W_2$, with $W_1$ and $W_2$ mutually independent, and let $\SIM$ and $\SA$ be defined respectively by \eqref{eq:GSBM}--\eqref{eq:LocSBM} and  \eqref{eq:GSBM}--\eqref{eq:IMSBM} with parameters $\sigma_\mathcal I, \mu_\mathcal I, \delta_\mathcal I: =2 \alpha-1 $ and $\sigma_\mathcal A, \mu_\mathcal A, \delta_\mathcal A$ respectively. Assume further $\delta_\mathcal I, \delta_\mathcal A \neq 0.$ 
Both models violate the NIP property, and thus admit an increasing arbitrage.
\end{thm}
\begin{proof}
    The violation of the structure condition for $\SIM$ was directly addressed in \cite{rossello2012arbitrage}, 
    Proposition 2. The same claim can be verified for $\SA$, in a similar fashion. We remove the subscript from parameters for ease of notation. 
Applying Tanaka's Lemma to $|W_2|$, we see that $\SA$ satisfies the following SDE
\begin{align}
    \SA_t &=S_0\exp\left( \mu t + \sigma \tilde W_t+\delta\sigma L^{W_2}_t \right)
\end{align} where \begin{equation}\label{eq:wtilda}\tilde W_t=\sqrt{1-\delta^2} W_{1,t} +\delta\int_0^t \text{sgn}(W_{2,u})d W_{2,u}\end{equation} is a Brownian motion by L\'evy's Theorem, since $\langle \tilde W \rangle_t=t$. Applying It\^o's Lemma now leads to
\begin{equation}\label{eq:ito} 
\SA_t=S_0+(\mu+\sigma^2/2)  \int_0^t   \SA_u d u +\sigma  \int_0^t \SA_u  d \tilde W_{u} + \sigma \delta \int_0^t \SA_ud L^{W_2}_u.
\end{equation}
Now write $\SA=F+M+S_0$ the It\^o-Doob decomposition of $\SA$ with $F=(F_t)_{t\geq 0}$  the finite variation part and $M=(M_t)_{t \geq 0}$ the local martingale part. We have 
\begin{equation}
    F_t=(\mu+\sigma^2/2)  \int_0^t   \SA_u d u +\sigma \delta \int_0^t \SA_u d L^{W_2}_u, \qquad \langle M \rangle_t=\sigma^2  \int_0^t (\SA_u)^2  d u  .
\end{equation}
  
    Let $\mathcal Z(\omega)=\{t \geq 0 \, | \, W_{2,t}(\omega)=0\}$.  It follows by \cite{protter2012stochastic}, Theorem IV.69, that supp\, $ dL^{W_2}=\mathcal Z$, and thus for all $T>0$ it holds that 
    \begin{equation}
    \int_0^T
I_{\mathcal Z}dF_u
=
\sigma\delta
\int_0^T
I_{\mathcal Z}
\SA_u dL_u^{W_2}
=
\sigma\delta
\int_0^T
\SA_u dL_u^{W_2}.
\end{equation}
implying $dF(\mathcal Z \cap [0,T])\neq 0$ a.s. But now since $d \langle M \rangle$ is absolutely continuous with respect to the Lebesgue measure, and Leb$(\mathcal Z )=0$ a.s., it also holds $d \langle M \rangle(\mathcal Z \cap [0,T])=0$ a.s., so that
    $dF \not \ll  d \langle  M \rangle$ and the structure condition fails. By
\cite{fontana2015weak}, Theorem 3.1 and the ensuing discussion, the NIP condition fails as well.
   \end{proof}
We can already conclude that arbitrage violations are very serious in GSBM models. In order to clear the field of all possible remaining doubts, we now show the failure of NIP directly, that is, by explicitly exhibiting an increasing arbitrage. The example is constructed in a way similar to \cite{fontana2015weak}, Example 7.1, and \cite{buckner2024arbitrage} in the case of a reflecting model. In case of a reflection the arbitrage is apparent: buy (sell) if and when the price hits the minimum (maximum), and since there is only one way it could go next, sell (buy back) at a profit immediately afterwards. In the case of an SBM   the intuition is more nuanced. However, theoretically, what really is at play is not so much the presence of a reflection, but rather the complementary nature of Brownian and local time measures in the SBM dynamics, which can be exploited by one-touch strategies. These produce zero gains on a Brownian path, but are instead carried by the local time support. 

\begin{prop}\label{prop:arbitrage} 
In the setup of Theorem \ref{thm:noarbthm}, the following holds:
\begin{itemize} 
\item[(i)]
 the process $H^{\mathcal A}=(H^{\mathcal A}_t)_{t \geq 0}$ given by $H^{\mathcal A}_t= \frac{ \mbox{{\upshape sgn}}(\delta_\mathcal A)}{\SA_t} I_{\{t \in \mathcal Z_\mathcal A\}}$, and $\mathcal Z_\mathcal{A}$  the random set  $\mathcal Z_\mathcal A(\omega)=\{t >0 \, | \, W_{2,t}(\omega)=0 \}$, is an increasing profit for $\SA$;
\item[(ii)] the process $H
^{\mathcal{I}}=(H^{\mathcal{I}}_t)_{t \geq 0}$  given by $H^{\mathcal{I}}_t=  \mbox{{\upshape sgn}}(\delta_\mathcal I)e^{-\mu t-\log S_0}I_{\{t \in \mathcal Z_{\mathcal{I}}\}}$, and  $\mathcal Z_{\mathcal{I}}(\omega)=\{t >0 \, | \, \log(\SIM_t(\omega)/S_0e^{\mu t})=0 \}$, is an increasing profit for $\SIM$.
\end{itemize}
\end{prop}
\begin{proof}
\emph{(i)}. Since $(\SA, W_{2})$ is continuous and $\mathcal F$-adapted, it is also predictable;  
therefore $H^\mathcal A$ is predictable since it is obtained by applying a measurable function (an inverse function multiplied by the indicator function of the level set of a measurable function) to a predictable process.
Using \eqref{eq:ito} the gains process from $H^\mathcal A$ is given by
\begin{align}\label{eq:gains}
    G^\mathcal A_t:=\int_0^t H^\mathcal A_u d \SA_u= \text{sgn}(\delta_\mathcal A)   & \Bigg((\mu_\mathcal A + \sigma_\mathcal A^2/2)  \int_0^t   I_{\mathcal Z_\mathcal A} d u +\sigma_\mathcal A  \int_0^t I_{\mathcal Z_\mathcal A}  d \tilde W_{u} + \sigma_\mathcal A \delta_\mathcal A \int_0^t I_{\mathcal Z_\mathcal A} d L^{W_2}_u \Bigg)
\end{align}
and since, 
as already observed, by \cite{protter2012stochastic}, Theorem IV.69,  supp $d L_t^{W_2}=\mathcal \{t>0 \, | \, W_{2,t}=0\} = \mathcal Z_\mathcal A$. Since  Leb$(\mathcal Z_\mathcal A)=0$ a.s. the first pathwise Riemann integral in \eqref{eq:gains} equals zero, and so does the stochastic integral in $d \tilde W$, because if $D$ is any square-integrable process supported in a zero Lebesgue measure set a.s., then by the It\^o isometry $\E[(\int_0^tD_u d W_u)^2 ]=\E[\int_0^tD_u^2 d u]=0$ entailing $\int_0^tD_u d W_u=0$.  
Therefore  \eqref{eq:gains} reduces to
\begin{align}\label{eq:hloc}
    G^\mathcal A_t=   \sigma_\mathcal A   |\delta_\mathcal A|L^{W_2}_t.
\end{align}
 The integrability conditions on $H^\mathcal A$  are thus in particular satisfied and the strategy is admissible.
 
Now, the prefactor in \eqref{eq:hloc} is positive, and the local time is well-known to be increasing: one way of seeing this is to observe that  for all $\epsilon>0$ the expressions on the right hand of $\eqref{eq:ltime}$ are increasing, and monotonicity is preserved by the pointwise limit. Finally  by a theorem of P. L\'evy (\cite{bjork2019pedestrian}, Proposition 3.5), $L^{W_2}_t$ has half-normal distribution. Since such law  is absolutely continuous, it follows $\bp(
    G^\mathcal A_t>0)>0$, for all $t>0$, in particular showing that $H$ is an increasing profit.
 \emph{(ii)}. The strategy $H^{\mathcal{I}}$ is obviously adapted, and since from \eqref{eq:GSBM}--\eqref{eq:LocSBM} it holds that $\mathcal Z=\{t \geq 0 \, | \,\XIM_t=0 \}=\{t \geq 0 \, | \,\SIM_t=S_0 e^{\mu t}\}$,  we have, similar to \eqref{eq:gains}, that
\begin{align}\label{eq:gains2}
    G^{\mathcal{I}}_t&:=\int_0^t H^{\mathcal{I}}_u d \SIM_u\nonumber \\ &  = \text{sgn}(\delta_\mathcal I)  \Bigg((\mu_\mathcal I + \sigma_\mathcal I^2/2)  \int_0^t   I_{\mathcal Z_{\mathcal{I}}} d u +\sigma_\mathcal I  \int_0^t   I_{\mathcal Z_{\mathcal{I}}} d  W_{u} + \sigma_\mathcal I \delta_\mathcal I \int_0^t I_{\mathcal Z_\mathcal{I}} d L^{\XIM}_u \Bigg).
\end{align}
The first two integrals vanish again, and again the third is taken on the support of $d L^{\XIM}$ and the strategy is admissible. Now it suffices to observe that $L^{\XIM}_t \stackrel{d}{=}L^{W}_t$ (e.g. \cite{rossello2012arbitrage}, p. 51) so that $G^{\mathcal{I}}_t$ matches in distribution \eqref{eq:hloc}, and the claim follows using the same arguments as in \emph{(i)}.
\end{proof}


Comparing  points \emph{(i)} and \emph{(ii)} of Proposition \ref{prop:arbitrage}   the difference between the arbitrage strategies for $\SA$ and $\SIM$ is that the former requires market information beyond the asset price,  whereas the latter can be implemented by observing $S$ alone. This is an effect of the different Markov properties of the two processes.

In the next section we shall show that even though the AGSBM model violates no-arbitrage, a continuous Markovian martingale with same marginals as $\SA$ exists, and it is a classic one-dimensional Brownian integral: financially, a local volatility model.



\section{The skew-normal option pricing formula}\label{sec:ZhuHeOP} 


In this section we first calculate from scratch the call payoff expectation of the skew-normal marginal laws, which overcomes some new errors  present in the equation of \cite{zhu2018new} formula correcting that of \cite{corns2007skew} (see  Appendix B). We then proceed to show that such a formula does not indeed admit arbitrage.

A no-arbitrage  call option price $C(T ,K, S, r)$ with strike $K$, maturity $T$ written on an asset of current price $S$ and for a prevailing risk-free rate on the market, is a function $C :D \subseteq \R_+ \times \R \times \R \times [0,\infty) \rightarrow \R_+ $, with $D$ a connected subset of $\R_+$, such that
\begin{enumerate}
    \item[A1.] $C(\cdot, K, S, r)$ is continuous and  increasing for all $K$, $S$, $r$;
    \item[A2.] $C(T, \cdot, S, r)$ is convex for all $T,S, r$;
    \item[A3.] $\lim_{K \rightarrow \infty} C(T,K,S,r)=0$ for all $T,S,r$;
    \item [A4.] $\lim_{T \rightarrow 0}C(T,K,S,r)=(S-K)^+$ for all $K, S$;
    \item[A5.]  $(S-Ke^{-r T})^+ \leq C(T,K,S, r) \leq  S$ for all $T, K, S$.
\end{enumerate}
 This sufficient, and close to necessary, characterization is taken from  \cite{rop:10}, extended to the case of possibly positive rates. Several equivalent definitions exist, in slightly different set-ups: see \cite{lowther2008fitting},
 \cite{dav+hob:07}, \cite{carr2005note}. Any of these sets of assumptions imply that $S_T$ is increasing in the convex order, and hence  by \cite{kellerer1972markov} that a  process  $S=(S_t)_{t \geq 0}$ having law $S_t$, and such that its discounted process $\tilde S=(e^{-r t} S_t)_{t \geq 0}$ is a martingale exists.

As in recent option pricing literature (\cite{azz+tor:23}, \cite{car+tor:21}),
 we turn the modeling paradigm around. Instead of deducing the option prices from a model, we start from a set of no option prices, try to ascertain whether
they are free of arbitrage, and then construct a martingale supporting them.

   %

For the rest of the paper we will only discuss the AGSBM model and thus for notational convenience drop the superscript $\mathcal A$ from the asset price dynamics.

Let $\Phi$ and $\phi$ be respectively the standard normal cumulative distribution function (CDF) and probability distribution function (PDF). Introduce two-parameter extended-skew normal as the absolutely continuous distribution with PDF (\cite{azzalini2013skew}, Section 2.2)
given by \begin{equation}\label{eq:azzskew2} \phi(x; \alpha, \tau)=\frac{\phi(x)\Phi\left(\tau\sqrt{1+\alpha^2} + \alpha x\right)}{\Phi(\tau)}, \qquad x, \alpha, \tau \in \mathbb R, \end{equation}
and denote the corresponding CDF by $\Phi(x; \alpha, \tau)$. The classic one-parameter   skew normal is obtained as the $\tau=0$ case of \eqref{eq:azzskew2}, that is $\phi(x; \alpha):=\phi(x;\alpha, 0)= 2 \phi(x)\Phi(\alpha x)$, and the corresponding CDF is thus $\Phi(x;\alpha):=\Phi(x; \alpha,0)$.  The two-parameter and classic skew normal distributions are  denoted respectively  SN$(\alpha)$, SN$(\alpha,\tau)$. For $\alpha=0$ we have the normal N$(0,1)$ distribution. Scale-location shifts of the form $\sigma X+ \mu$, with $X \sim$ SN$(\alpha, \tau)$, SN$(\alpha)$ have distributions in the extended families, respectively, SN$(\alpha, \tau, \mu, \sigma)$ and SN$(\alpha, \mu, \sigma)$.  A number of properties of the class SN$(\alpha, \tau, \sigma, \mu)$, some of which were previuosly not available to the best of our knowledge, are found in Appendix A.

We define the \emph{skew-normal call function} to be the  expectation of the call option payoff calculated on the family on $T$ of the marginals  of a AGSBM, which, as noted, are in the SN class. The price distribution is obtained by exponentiating a scale-location family $(X_t)_{t \geq 0}$ of skew normal random variables plus the risk-free rate and a location change ensuring that the resulting random variables have mean equal to the forward price $S_0 e^{rt}$.
In other words, for $S_0, r, t>0$, $\delta \in (-1,1)$  under some measure $\mathbb Q$ on the space $(\Omega, \mathcal F, \mathcal F_{t \geq 0}, \mathbb Q)$ we consider the following random variables
\begin{align}\label{eq:SNvars}
    S_t&=S_0 e^{r t + X_t}, \nonumber \\ X_t &\sim \mbox{SN}\left(\frac{\delta}{\sqrt{1-\delta^2}}, \sigma \sqrt{t},-\mu^*(t,  \sigma, \delta) \right), \nonumber \\
      \mu^*(t, \sigma, \delta)&=\sigma^2 t/2 +\log\left( 2 \Phi( \delta \sigma  \sqrt{t}  )\right).
\end{align}
so that $X_t$ has PDF $p(t, \cdot)$ given as
\begin{equation}\label{eq:Xtpdf}
    p(t,x)=\frac{1}{\sigma \sqrt{t}}\phi\left(\frac{x+\mu^*(t)}{\sigma \sqrt{t}}; \frac{\delta}{\sqrt{1-\delta^2}} \right).
\end{equation}
With hindsight, $\mathbb Q$ will be a martingale measure for the process supporting the random variables above. Denoting still $\mathbb E[\cdot]$ the  $\mathbb Q$-expectation, that $\mathbb E[e^{-r t}S_t]=S_0$ for all $t \geq 0$, follows directly from the moment generating function in \eqref{eq:snmgf}. 



We begin by defining the set of putative option prices, as the discounted call option payoff expectations in the laws $(S_t)_{t \geq 0}$.

\begin{prop}\label{prop:rightcall}
    Let $(S_t)_{t \geq 0}$ be the family of random variables in \eqref{eq:SNvars}, and let $C(T,K,S_0,r):=\E[e^{- r T}(S_T-K)^+ ]$. We have that 
\begin{align}\label{eq:call_m0}
C(T,K,S_0,r)& =S_0 \left(1-\Phi\left(b_+\left(\log(K/S_0),\sigma \sqrt{T}, \delta, r T\right); \frac{\delta}{\sqrt{1-\delta^2}}, \delta\sigma\sqrt{T}   \right)\right) \nonumber \\ & \phantom{xxxxxxxxxxxx} -Ke^{-r T}\left(1-\Phi\left(b_-\left(\log(K/S_0),\sigma \sqrt{T}, \delta, r T\right);\frac{\delta}{\sqrt{1-\delta^2}} \right) \right)
\end{align}
with
\begin{align}\label{eq:b}
b_\pm( x, y, d, \rho)=\frac{x+\log\left(2  \Phi(d\, y) \right)- \rho \mp y^2/2 }{y }, \qquad x \in \R, y >0, r\geq 0, \delta \in (-1,1).
 \end{align}

\end{prop}

\begin{proof}
The calculation is classic and straightforward. Using \eqref{eq:Xtpdf}
have that 
\begin{align}
&\mathbb E[e^{-rT}(S_T-K)^+]=
S_0\E[e^{X_T}1_{\{X_T\geq \log(K/S_0)-r T  \} }] -K e^{-rT} \mathbb Q(X_T\geq \log(K/S_0)-r T ) \nonumber \\
=&S_0\int_{ \log(K/S_0)-r T } ^\infty \frac{e^{x}}{\sigma \sqrt{T}}\phi\left(\frac{x+\mu^*(T)}{\sigma \sqrt{T}}; \frac{\delta}{\sqrt{1-\delta^2}} \right)dx  \nonumber \\ &\phantom{xxxxxxxxxxxxxxxxxxxx}
-K e^{-r T} \int_{ \log(K/S_0)-r T } ^\infty \phi\left(\frac{x+\mu^*(T)}{\sigma \sqrt{T}}; \frac{\delta}{\sqrt{1-\delta^2}} \right)\frac{dx}{\sigma{\sqrt{T}}}\nonumber \\
=&S_0\int_{  \log(K/S_0)-r T} ^\infty \phi\left( \frac{x - \sigma^2 T +\mu^*(T)}{\sigma \sqrt{T}}; \frac{\delta}{\sqrt{1-\delta^2}}, \delta \sigma \sqrt{T} \right) \frac{dx}{\sigma \sqrt{T}} 
\nonumber \\ &\phantom{xxxxxxxxxxxxxxxxx}-K e^{-r T} \left( 1- \Phi\left(b_-(\log(K/S_0), \sigma \sqrt{T}, \delta,r \, T); \frac{\delta}{\sqrt{1-\delta^2}} \right) \right). \label{eq:lastline}
\end{align}
In the last line, exploiting $\mathbb E[e^{X_t}]=1$, we applied Proposition \ref{prop:esscher} to the first integral term, under the specification $\theta=1$, $\tau=0$, $\sigma \sqrt{T}$ in place of $\sigma$, and $\alpha=\delta/\sqrt{1-\delta^2}$, so that  $\alpha/\sqrt{1+\alpha^2}= \delta$. The further substitution $y= (x - \sigma^2 T +\mu^*(T))/\sigma \sqrt{T}$ in \eqref{eq:lastline}, now yields \eqref{eq:call_m0}--\eqref{eq:b}.

\end{proof}


This probabilistic expression is classic treatment. In particular we see that the emergence of the extended-skew normal law in the first summand as the familiar interpretation of dynamically changing measure from the risk-neutral one to the \emph{share measure}, i.e. the measure $\mathbb S$ on $(\Omega,\mathcal F_t, \mathcal F)$ whose martingale density is given by

\begin{equation}\label{eq:shareM}
    \frac{d \mathbb S}{d \mathbb Q}\Bigg |_{\mathcal F_t}= \frac{e^{-r t}S_t}{S_0}
\end{equation}
Since changing of measure using exponential martingale dynamics such as \eqref{eq:shareM} impacts the underlying  distribution tilting, in the skew-normal world Proposition \ref{prop:esscher} is the tool we need. 
As remarked in the Appendix,  when $\delta \neq 0$, this change of probability places the share measure outside the class of the risk-neutral one -- turning it from skew-normal to extended skew normal -- which is unlike the Black-Scholes case ($\delta=0$) when both expressions are normal, although with different expectations.

The argument  of Proposition \ref{prop:rightcall} can be repeated to establish that the put option  putative value given by $P(T,K,S_0,r):=\E[e^{- r T}(K-S_T)^+ ]$ satisfies
\begin{align}\label{eq:put_m0}
P(T,K,S_0,r)
&=
K e^{-rT}
\Phi\left(
b_-\left(\log(K/S_0),\sigma\sqrt{T},\delta,rT\right);
\frac{\delta}{\sqrt{1-\delta^2}}
\right)
\nonumber\\
&\quad
-
S_0
\Phi\left(
b_+\left(\log(K/S_0),\sigma\sqrt{T},\delta,rT\right);
\frac{\delta}{\sqrt{1-\delta^2}},
\delta\sigma\sqrt{T}
\right).
\end{align}
Note that the difference of \eqref{eq:call_m0} and \eqref{eq:put_m0} establishes, at least formally, a put-call parity relation, that will be fully justified once martingale dynamics for with marginals $(S_t)_{t \geq 0}$ are shown to exists.

As explained, in order to show that $C$ is free of arbitrage, and that a supporting martingale for the option price exists, we must prove properties A1 to A5. We make two separate mathematical statements: one for A2--A5, and one for A1. As part of the latter we explicitly exhibit the maturity derivative of $C$, which is instrumental to determination of the local volatility coefficient later on.

\begin{prop}\label{prop:noarb12}
The skew-normal call price function $C(K,T,S_0,r)$ satisfies  properties A2 to A5. 
\end{prop}
\begin{proof}
    The properties are obvious and valid in general once a clear probabilistic representation of prices through non-negative random variables is provided as a starting point. Call $f_t$ the law of $S_t$ in \eqref{eq:SNvars}. The \cite{breeden1978prices} remark leads to $e^{r T}\partial_{KK} C(S_0,K,T,r)=f_T(K) \geq 0$ implying A2. For A3, the property follows immediately taking the limit on $K$ in \eqref{eq:call_m0}. 
    Regarding A4, the O/ITM cases  are an immediate consequence of the fact that for $K<S_0$ we have $b_\pm(\log(K/ S_0), \sigma \sqrt{ T},\delta, r T) \rightarrow -\infty$, while if $K>S_0$ $b_\pm(\log(K/ S_0), \sigma \sqrt{ T},\delta, r T) \rightarrow \infty$. For $K=S_0$ observe, by a Taylor expansion, that as $T \rightarrow 0$ it holds \begin{equation}
        2 \Phi( \delta \sigma \sqrt{T})) = 1 +\sqrt{\frac{2 }{ \pi}}\delta \sigma\sqrt{T} +O(T).
    \end{equation}
Whence, using $\log (1+x) \sim x$, it follows
\begin{align}\label{eq:logtaylor}
b_\pm(0, \sigma \sqrt{T},\delta, r)=&\frac{\log(2 \Phi(\delta \sigma \sqrt{T}))-(r \pm \sigma^2/2)T }{\sigma \sqrt{T}}  = \frac{\log(1+ \sqrt{\frac{2 }{ \pi}}\delta \sigma\sqrt{T}+ O(\sqrt{T})  )}{\sigma\sqrt{T}}  + O(\sqrt{T})\nonumber \\ &\rightarrow \sqrt{\frac{2 }{ \pi}}\delta, \qquad T \rightarrow 0.
\end{align}
Now for the second term in \eqref{eq:call_m0} we have, that, 
by continuity,  as $T \rightarrow 0$ 
\begin{equation}
S_0e^{-r T}\left(1-\Phi\left(b_-\left(0,\sigma \sqrt{T}, \delta, r\right);\frac{\delta}{\sqrt{1-\delta^2}} \right) \right) \rightarrow S_0 \left(1-\Phi\left( \delta \sqrt{\frac{2}{\pi}};\frac{\delta}{\sqrt{1-\delta^2}} \right) \right).
\end{equation} So for the ATM call price to tend to the intrinsic value 0 as time to maturity vanishes  it is sufficient to show that  
\begin{equation}\label{eq:firtTlimit}
\Phi\left(b_+\left(0,\sigma \sqrt{T}, \delta, r\right);\frac{\delta}{\sqrt{1-\delta^2}}, \delta \sigma  \sqrt{T}  \right) \rightarrow  \Phi\left( \delta \sqrt{\frac{2}{\pi}};\frac{\delta}{\sqrt{1-\delta^2}} \right), \qquad T \rightarrow 0.
\end{equation}
which requires taking the limit on $T$ inside an integral. As pointwise it holds  $\phi(x; \alpha, \tau) \rightarrow \phi(x;\alpha)$ as $\tau \rightarrow 0$, so for \eqref{eq:firtTlimit} to hold dominated convergence must apply. But this is clearly the case since $\phi(x; \alpha, \tau) < 2 \phi(x)$ , for all $\alpha, \tau, x$, and A4 also follows  for the ATM case.

    Finally, the upper bound of A5 is obvious. Being $S_t$ a non negative random variable and $(x-k)^+$, $x,k >0$ a non-negative function, its expectation must be non negative, so one only requires to show $C(T,K,S_0,r) \geq S_0-Ke^{-r T}$. But with $ P(T,K,S_0,r)$  in \eqref{eq:put_m0}, which is also positive for the same reasons as $C$, we have 
    \begin{equation} C(T,K,S_0,r)= S_0-K e^{-r T}+ P(T,K,S_0,r)
    >S_0-K e^{-r T}\end{equation}
showing A5.
\end{proof}

 Proving A1 is slightly more challenging. In order to do so the properties of the  extended SN law uncovered in the Appendix  come  into play.

\begin{thm}\label{thm:noarb3}
Let $C$ given by \eqref{eq:call_m0} be the skew-normal putative call price function. Then 
\begin{align}\label{eq:CpartialT2}
& \quad \partial_T C(T,K,S_0,r)=rKe^{-rT}\biggl(1-\Phi\biggl(b_-(\log(K/S_0,\sigma \sqrt{T}, \delta, r T)
;\frac{\delta}{\sqrt{1-\delta^2}}\biggr)\biggr) \nonumber \\ &+  S_0 \frac{\sigma}{2\sqrt{T}}  \left[ \phi\left(b_+(\log(K/S_0),\sigma \sqrt{T}, \delta, r T);\frac{\delta}{\sqrt{1-\delta^2}},  \sigma \sqrt{T}\delta\right) - \right. \nonumber \\ &\phantom{xxxxxxxxxxxxxxxxxxx}\left.\delta\frac{\phi(\sigma \sqrt{T}\delta)}{\Phi(\sigma \sqrt{T}\delta)}\Delta \left(b_+(\log(K/S_0),\sigma \sqrt{T}, \delta, r T), \frac{\delta}{\sqrt{1-\delta^2}}, \sigma \sqrt{T}\delta \right)\right],
 \end{align}
with $\Delta$ defined in \eqref{eq:Delta}. Furthermore,  $\partial_T C\geq 0$ on its whole domain. 
As a consequence, $C(\cdot, K, S_0,r)$ 
is increasing in $T$ for all $K, S_0 >0, r \geq 0$.  
\end{thm}

\begin{proof}
Let $x(T)=b_+(\log(K/S_0,\sigma \sqrt{T}, \delta, r T)$, so that $b_-(\log(K/S_0),\sigma \sqrt{T}, \delta, r T)= x +\sigma \sqrt{T}$. We introduce the variable change
\begin{equation}\label{eq:subs}
    \alpha=\frac{\delta}{\sqrt{1-\delta^2}}, \qquad \tau=\delta \sigma \sqrt{T},
\end{equation}
using which we can rewrite \eqref{eq:call_m0}, with abuse of notation,  as follows
\begin{equation}\label{call_m0_2}
C(x, T, r)=S_0(1-\Phi(x;\alpha,\tau))-Ke^{-rT}(1-\Phi( x+\sigma \sqrt{T};\alpha)).
\end{equation}
By differentiating the function $C$ in \eqref{call_m0_2} with respect to $T$, this leads to
\begin{align}
\partial_T C(x,T,r)&=-S_0\bigl( \partial_x \Phi(x;\alpha,\tau) \partial_T x+\partial_{\tau} \Phi(x;\alpha,\tau)\partial_T \tau \bigr)+
rKe^{-rT}(1-\Phi(x+\sigma\sqrt{T};\alpha))\nonumber \\ &+Ke^{-rT}\partial_x \Phi(x+\sigma\sqrt{T};\alpha)\partial_T(x+\sigma\sqrt{T});
\end{align}
where
\begin{align}
\partial_T(x+\sigma\sqrt{T})=\partial_T x+\frac{\sigma}{2\sqrt{T}}, \quad \partial_T \tau=\frac{\delta \sigma }{2\sqrt{T}}, \quad \partial_x \Phi(x;\alpha,\tau)=\phi(x;\alpha,\tau).
\end{align}
Using Lemma \ref{lem:skewazzdertau} 
 we thus obtain
\begin{align}\label{eq:stop}
\partial_T C(x, T,r)&=-S_0\biggl[\phi(x;\alpha,\tau)\partial_Tx +\frac{\phi(\tau)}{\Phi(\tau)}\Delta(x;\alpha,\tau) \frac{\delta \sigma }{2\sqrt{T}}\biggr]+
rKe^{-rT}(1-\Phi(x+\sigma\sqrt{T};\alpha))\nonumber \\ &+Ke^{-rT}\phi(x+\sigma\sqrt{T};\alpha)\biggl(\partial_T x+\frac{\sigma}{2\sqrt{T}}\biggr).
\end{align}
Now recalling $x=b_+(\log(K/S_0), \sigma \sqrt{T}, \delta, r T)$ it is easy to check that  we have the Vega-type relation
\begin{equation}
    S_0\phi(x)=2\Phi(\tau)K e^{-r T} \phi(x+\sigma \sqrt{T})
\end{equation}
and besides, recalling \eqref{eq:subs} it holds
\begin{equation}
\Phi(\alpha(x+\sigma \sqrt{T}))=\Phi(\alpha x +\alpha \sigma \sqrt{T})=\Phi(\alpha x + \sqrt{1+\alpha^2}\tau),
\end{equation}
that in turn leads to
\begin{equation}\label{eq:rel_01}
    S_0 \phi(x,\alpha, \tau)=K e^{-r T}\phi(x+\sigma \sqrt{T}; \alpha).
\end{equation}
Applying \eqref{eq:rel_01}, the term $\partial_T x(Ke^{-rT}\phi(x+\sigma\sqrt{T};\alpha)-S_0\phi(x;\alpha,\tau))$ in \eqref{eq:stop} cancels, and we further have
\begin{align}
\partial_T C(x, T, r)&=-S_0\frac{\phi(\tau)}{\Phi(\tau)}\Delta(x,\alpha,\tau)\frac{\delta \sigma }{2\sqrt{T}} +
rKe^{-rT}(1-\Phi(x+\sigma\sqrt{T};\alpha))+S_0\phi(x;\alpha,\tau)\frac{ \sigma}{2\sqrt{T}} \nonumber \\ &= S_0\frac{\sigma}{2 \sqrt{T}}\left(\phi(x; \alpha,\tau)-\delta \frac{\phi(\tau)}{\Phi(\tau)}\Delta(x,\alpha, \tau)\right)+ 
rKe^{-rT}(1-\Phi(x+\sigma\sqrt{T};\alpha)).
\end{align}
After unwinding the change of variables, this is exactly \eqref{eq:CpartialT2}.
Finally,  observing that  $\delta=\alpha/\sqrt{1+\alpha^2}$ and applying Lemma \ref{lem:boundbydelta}, we obtain 
\begin{equation}
    S_0\frac{\sigma}{2 \sqrt{T}}\left(\phi(x; \alpha,\tau)-\delta \frac{\phi(\tau)}{\Phi(\tau)}\Delta(x,\alpha, \tau)\right) \geq S_0\frac{\sigma}{2 \sqrt{T}}\phi(x; \alpha,\tau) \left(1-\delta^2\right) \geq 0
\end{equation}
and since further the first line in \eqref{eq:CpartialT2} is obviously positive, the positivity of $C_T$ is established. The last claim is clear.

\end{proof}

 From Proposition \ref{prop:noarb12} and Theorem \ref{thm:noarb3} we conclude that A1-A5 all hold, and we can finally state the no-arbitrage property of the skew-normal call price function, in the sense of the existence of the existence of a probability space and a martingale whose discounted values supports the prices $C(T,K,S_0, r)$.

\begin{cor} There exists a  process $S=(S_t)_{t \geq 0}$, with $S_t$ distributed as \eqref{eq:SNvars}, and such that $\tilde S=(e^{-rt}S_t)_{t \geq 0}$ is a martingale, satisfying \eqref{eq:call_m0} and \eqref{eq:put_m0}. \end{cor}

In the next section we seek to associate to the SNLV option prices a positive continuous supporting strongly Markovian process.

\section{The skew-normal local volatility model}\label{sec:localvol}

Once a no-arbitrage call function (or put) is given, its local volatility can be formally computed. This rises two questions: first, whether the corresponding SDE is well-posed; second, whether the law of such a solution reproduces the prescribed call prices.

From this section on we assume that $r=0$. This substantially simplifies the formulae and it is of course not restrictive. Similar formulae, without changes in the interpretation of results, can be derived for positive risk-free rates, or alternatively $S$ can be understood in the following as a forward price model.

For a given no-arbitrage call option price function $C(T,K,S_0)$ let $c(\cdot,\cdot)$ be the call price process in moneyness coordinates,  given by the change of variable 
\begin{equation}\label{eq:cC} c(T, \log(K/S_0))=C(T,K,S_0), \qquad c(t,\kappa)=C(t,e^{\kappa},1)\end{equation}
Fix a filtered probability space $(\Omega, \mathcal F, \mathbb P)$ supporting a standard Brownian motion $W=(W_t)_{t \geq 0}$.  A \emph{local volatility model} is a solution, if it exists unique, of the SDE 
\begin{equation}\label{eq:SDE}
d S_t=S_t\sqrt{ 2 v_L(t, \log(S_t/S_0))} d W_t
\end{equation}
with
\begin{equation}\label{eq:v}
v_L(t,\kappa) =\frac{\partial c}{\partial t } (t, \kappa)\left(\frac{\partial^2 c}{\partial \kappa^2 } (t, \kappa)-\frac{\partial c}{\partial \kappa} (t, \kappa)\right)^{-1}, \quad \kappa \in \R, \quad t >0.
\end{equation}
The function $v_L$ is the \emph{local half-variance}. If \begin{equation}p(t,\kappa)=e^{-\kappa}\left(\frac{\partial^2 c}{\partial \kappa^2}(t,\kappa)-\frac{\partial c}{\partial \kappa}(t,\kappa)\right) \end{equation}
it is easy to show that $p \in C^{1,2}((0,\infty) \times \mathbb R))$ solves the Fokker Planck equation
\begin{equation}\label{eq:pFP}
\frac{ \partial p}{\partial t}(t, \kappa)= 
 \frac{ \partial^2 }{\partial \kappa^2}(v_L(t,\kappa)  p(t, \kappa))+\frac{ \partial }{\partial \kappa}(v_L(t,\kappa)  p(t, \kappa)), \qquad p(0,\kappa)=\delta_0(d \kappa).
\end{equation}
Also $p$ is a probability density since clearly
\begin{equation}\label{eq:eq:}
 p(t,x)= S_0e^x\frac{\partial^2 C}{\partial K^2}(t,S_0e^x,S_0),  
 \end{equation}
and as well-known $\partial^2 C/\partial K^2$
is the market-implied risk-neutral density
(\cite{breeden1978prices}).
Furthermore letting
 \begin{equation}
    \label{eq:p}p^*(t,\kappa)=e^{\kappa }p(t,\kappa)=\left(\frac{\partial^2 c}{\partial \kappa^2}(t,\kappa)-\frac{\partial c}{\partial \kappa}(t,\kappa) \right)\end{equation}
we have that the denominator $p^*$ of the local variance is the market-price density in the share measure. Of course that the probability laws $p$ and $p^*$ coincide with the law of $\log (S_t/S_0)$ respectively in the risk-neutral and share measures, with $S$ now a solution to \eqref{eq:SDE}--\eqref{eq:vSN}, should be made precise, as there is no a priori reason for which this should holds. For example, a second solution solution $\bar p$ to \eqref{eq:pFP} may exist, and the law of a (weak or strong) solution for \eqref{eq:SDE}--\eqref{eq:vSN}, even if it exists unique, may be $\bar p$ instead of  $p$. An argument must be made. For details on all the above see \cite{dupire1994pricing}, \cite{berestycki2002asymptotics}, \cite{carmona2009local}, \cite{torricelli2025parametric}.  

In this setting, we take $C$ to be the skew normal no arbitrage call option price, and proceed to structure a local volatility model with skew normal option prices/marginals.  Since both $p^*$ and $\partial c/\partial t$ are known from the previous sections, it is easy to compute the skew-normal half variance. Answering the two questions above (well-posedness of the SDE and whether it matches option prices) is easier we could show that we are in the standard setup of bounded and uniformly elliptic coefficients. Establishing this  in turns requires the identification of the time zero asymptotics of $v_L$. We have the following result, uncovering a useful probabilistic representation of the SDE coefficients.

\begin{prop}\label{prop:extvl} Let $C(T, K, S_0)$ be the zero-interest rate skew-normal no-arbitrage option pricing formula with underlying laws $(S_t)_{t \geq 0}$ distributed as \eqref{eq:SNvars}. Then the skew-normal local half-variance is given by
\begin{align}\label{eq:vSN}
     &v_L(t,\kappa)=  
\frac{\sigma^2}{2}\left[
1-\delta  R\left(b_+(\kappa,\sigma\sqrt{t},\delta,0), 
\frac{\delta}{\sqrt{1-\delta^2}},
\delta\sigma\sqrt{t}\right)\right]
\end{align}
with
\begin{equation}\label{eq:R}
R(x,\alpha, \tau):=
\,\frac{\phi(\tau)}
{\Phi(\tau)} \frac{\Delta\left(
x,\alpha,\tau \right)}{\phi(x,\alpha,\tau)}=\frac{\alpha}{\sqrt{1+\alpha^2}}\frac{1}
{m(-\tau)}\mathbb E\left[
{m(U_{x,\tau})}\right], \quad x, \alpha,  \textcolor{blue}{\tau}\in \mathbb R,
\end{equation}
with $m$ the Mills ratio and $U_{x,\tau}$ as in Lemma \ref{lem:boundbydelta}.

The function $v_L$ can be extended to $[0,\infty) \times \mathbb R$ by letting
\begin{equation}\label{eq:lim_T0}
v_L(0,\kappa)
=
\begin{cases}
\displaystyle{\sigma^2\left(1-\delta^2\right)/2},& \kappa\delta<0,\\
\sigma^2/2,& \kappa\delta>0,
\end{cases}
\end{equation}
 and 
\begin{equation}\label{eq:vl00}
v_L(0,0)
=
\frac{\sigma^2}{2}
\left[
1-
\delta
R\left(
\delta\sqrt{\frac{2}{\pi}},
\frac{\delta}{\sqrt{1-\delta^2}},0
\right)
\right].
\end{equation}

This extension is  differentiable outside $(0,0)$ which is a discontinuity point for $v_L$.
Furthermore $v_L$ is bounded with
\begin{equation}\label{eq:vlbounds}
\frac{\sigma^2}{2}(1-\delta^2) \leq v_L(t, \kappa) \leq \frac{\sigma^2}{2} ,
\end{equation}
for all $t\geq 0$, $\kappa \in \mathbb R$.
\end{prop}

\begin{proof}

 As was shown in Proposition \ref{prop:rightcall}, combining \eqref{eq:Xtpdf} and  Proposition \ref{prop:esscher} that under the share measure $X_t \sim$ SN$\left(\frac{\delta}{\sqrt{1-\delta^2}}, \delta \sigma \sqrt{t}, \sigma^2t-\mu^*({t}),\sigma \sqrt{t} \right)$ with PDF given by
\begin{equation}\label{eq:pstar}
p^*(t,\kappa)=\frac{1}{\sigma \sqrt{t}}\phi\left(b_+(\kappa, \sigma \sqrt{t},\delta,0); \frac{\delta}{\sqrt{1-\delta^2}}, \delta \sigma \sqrt{t} \right).
\end{equation} Clearly $\partial c(t,\kappa)/\partial t=\partial C(t,e^{\kappa},1)/\partial T$; therefore evaluating \eqref{eq:CpartialT2} in $K=e^{\kappa}$ and $r=0$, substituting it in \eqref{eq:v} together with \eqref{eq:pstar} and simplifying as necessary yields \eqref{eq:vSN}.

We now turn to showing \eqref{eq:lim_T0}.  Assume first $\kappa \neq 0$. We need to show that

\begin{equation}\label{eq:Rlimit}
 R\left(b_+(\kappa,\sigma\sqrt{T},\delta,0), 
\frac{\delta}{\sqrt{1-\delta^2}},
\delta\sigma\sqrt{T}\right)\rightarrow \left\{ \begin{array}{cc} \delta & \mbox{ if } \kappa \delta <0  \\  0 & \mbox{ if } \kappa \delta >0   \end{array} \right.
\end{equation}
when $T\rightarrow 0$. As already observed
$b_+(T,\kappa)  \sim \kappa/\sigma \sqrt{T}$ as $T \rightarrow 0$, so let , $s=\sigma \sqrt{T}$ $x(s)=\kappa/ s$, $\tau=\delta  s $ and the usual substitutions $\alpha=\delta/\sqrt{1-\delta^2}$, $\delta=\alpha/\sqrt{1+\alpha^2}$ then \eqref{eq:Rlimit} is implied by
\begin{equation}\label{eq:Rlimit2}
 R\left(x(s), \alpha,\delta s \right)\rightarrow \left\{ \begin{array}{cc} \delta & \mbox{ if } \kappa \delta <0  \\  0 & \mbox{ if } \kappa \delta >0   \end{array} \right.
\end{equation}
as $s \rightarrow 0$. Now if $k \delta <0$, i.e $x(s) \delta \rightarrow -\infty$ then $-\delta(s+x(s)) \rightarrow \infty$ whereas if $k \delta  >0$ then $x(s) \delta \rightarrow \infty$, and $-\delta(s+x(s)) \rightarrow -\infty$. 

Assume we are in the latter case.
Applying Lemma \ref{lem:boundbydelta} to $R$ to obtain the right hand side of \eqref{eq:R},
and a further substitution leads to
\begin{equation}\label{eq:Rtruncated}
    R(x(s), \alpha, \delta s)=\delta \frac{\phi(s \delta)}{\Phi(s \delta)}\frac{\int_{-\delta (s+ x(s))}^\infty \phi(u \sqrt{1+\alpha^2})m(u+\delta x(s))du }{\int_{-\delta(s+ x(s))}^\infty\phi(u \sqrt{1+\alpha^2})du}.
\end{equation}
It is easy to show that $\phi(z)/\Phi(z)$ is a continuous decreasing function, and so is $m$. Hence 
\begin{equation}
   \frac{\phi(s \delta)}{\Phi(s \delta)} \phi(u \sqrt{1+\alpha^2})m(u+\delta x(s))1_{\{u>\geq -\delta(s+x(s) \}} \leq \frac{\phi(0)}{\Phi(0)}\phi(u \sqrt{1+\alpha^2}) m(u) 
\end{equation}
But now for all $c>0$ it is $\int_c^\infty \phi(u \sqrt{1+\alpha^2}) m(u) < \infty$ because $m$ is bounded on $[c, \infty)$ and since $m(u) \sim \sqrt{2 \pi}e^{u^2/2}$ as $u \rightarrow -\infty$ then $m(u)\phi(u \sqrt{1+\alpha^2}) \sim e^{-x \alpha^2/2}$ and therefore we  also obtain that $\int_{-\infty}^c \phi(u \sqrt{1+\alpha^2}) m(u) < \infty$. In conclusion we can apply the dominated convergence theorem. Passing the limit on $s$ in the fraction \eqref{eq:Rtruncated} one sees that the denominator tends to 1 and the numerator to 0, proving the second line of \eqref{eq:Rlimit2}.


Conversely assume $\kappa \delta <0$, i.e. $\delta x(s) \rightarrow - \infty$. The asymptotic becomes a matter of probabilistic convergence. Let $u>0$ arbitrary;  using the truncated CDF and the usual asymptotics of $\Phi$ at $-\infty$  we see that 
\begin{align}\label{eq:0lawconv}
    &\mathbb P(U_{ x(s), \delta s}>u)=\frac{1-\Phi\left(\sqrt{1+\alpha^2}(u-\delta x(s))\right)}{1-\Phi\left(\sqrt{1+\alpha^2}(-\delta( s + x(s))) \right)}=  \frac{\Phi\left(\sqrt{1+\alpha^2}(-u+\delta x(s))\right)}{\Phi\left(\sqrt{1+\alpha^2}(\delta (s + x(s)) \right)} \nonumber \\ & \sim \frac{|\delta s + \delta x(s))|
    }{|-u +\delta x(s)|}\frac{\phi\left(\sqrt{1+\alpha^2}(-u+\delta x(s))\right)}{\phi(\sqrt{1+\alpha^2} \delta x(s))}\sim \frac{1}{\sqrt{2 \pi}}e^{\frac{1}{2}\left(-(1+\alpha^2)u +2 u \delta x(s) \right)}
\end{align}
implying $\mathbb P(U_{ x(s), \delta s} \geq u) \rightarrow 0$ if $u>0$ and $\mathbb P(U_{ x(s), \delta s} \geq 0) \rightarrow 1$, which establishes that  $U_{ x(s), \delta s}  \stackrel{d}{\rightarrow} 0$, as $s \rightarrow 0$.
Let now $s_0>0$ and $C \leq 0$ such that supp $U_{ x(s), \delta s} \subset [C,\infty) $ for all $s < s_0$ (one can take $C=0$  for any $s_0$ if $\delta >0$). The function $h(u)=m(u)\mathbf 1_{\{u \geq C\}}+m(C)\mathbf 1_{\{u <C\}}$ is a continuous bounded function, hence from the Portmanteau theorem we obtain, for all $s < s_0$, that \begin{equation} \mathbb E\left[m\left(U_{ x(s), \delta s}\right)\right]=\mathbb E\left[h\left(U_{ x(s), \delta s}\right)\right] \rightarrow m(0)\end{equation} as $ s \rightarrow 0$. In conclusion we have
\begin{equation}
    R\left(x(s), \alpha,\delta s \right)= \delta \frac{\phi(\delta s)}{\Phi(\delta s)}\mathbb E[m\left(U_{ x(s), \delta s}\right)] \rightarrow \delta \frac{\phi(0)}{\Phi(0)} m(0)=\delta
\end{equation}
and the first line of \eqref{eq:Rlimit2} is also proven.

The case $\kappa=0$ is a simple limit 
without indeterminacy, and the claim on extendibility is now obvious.

Equation \eqref{eq:vlbounds} follows immediately by using the right hand side of \eqref{eq:keybound} in \eqref{eq:vSN}--\eqref{eq:R}.
\end{proof}


The function $v_L$ is thus globally bounded and uniformly elliptic.  We shall show next that its spatial partial derivative diverges around the origin, uniformly in $x$, and in an integrable way.

\begin{prop}\label{lem:llipschitz}  For all $t \ge 0$, $\kappa \in \mathbb R$ it holds that 
\begin{equation}
    |    \partial_\kappa v_L(t, \kappa)| \leq \frac{\sigma \delta^2}{2\sqrt{1-\delta^2}}\frac{1}{\sqrt{t}}.
\end{equation}
\end{prop}

\begin{proof}
 Since $\partial _\kappa b_+(\kappa, \sigma \sqrt{T}, \delta,0)=(\sigma \sqrt{t})^{-1}$ by directly differentiating \eqref{eq:vSN}, and with the usual variable change \eqref{eq:subs}, it holds that
\begin{equation}\label{eq:VSNkappa}
    \partial_\kappa v_L(t, \kappa)=\textcolor{blue}{-}\frac{\sigma \delta}{2 \sqrt{t}}\partial_x R(x, \alpha, \tau)\Big|_{x=b_+(\kappa, \sigma \sqrt{t}, \delta, 0)}.
\end{equation}
Hence it is enough to show that $| \partial_x R(x, \alpha, \tau)  | \leq|\alpha|$. To this end we use the probabilistic representation of Lemma \ref{lem:boundbydelta}. The PDF $\phi^x$ of $U_{x, \tau}$  is given by \begin{equation}
\phi^x(u)=\frac{\phi\left( u \sqrt{1+\alpha^2} -\alpha x\right)}{1-\Phi(-\sqrt{1+\alpha^2}\tau-\alpha x)}\mathbf1_{\{u \geq -\tau\}}=\sqrt{1+\alpha^2}\frac{\phi\left(u \sqrt{1+\alpha^2} -\alpha x\right)}{\Phi(\sqrt{1+\alpha^2}\tau+\alpha x)}\mathbf1_{\{u \geq -\tau\}}.
\end{equation} Fix $K >0$ and let us restrict to $|x| \leq K$. Differentiating yields 
\begin{align}\label{eq:partialxtrunc}
    \partial_x \phi^x(u)&=\alpha \mathbf \phi^x(u) \left(u \sqrt{1+\alpha^2} -\alpha x -\frac{\phi(\sqrt{1+\alpha^2}\tau+\alpha x)}{\Phi(\sqrt{1+\alpha^2}\tau+\alpha x)} \right) \nonumber \\  & <|\alpha| \phi^x(u)(u \sqrt{1+\alpha^2}+K |\alpha|).\end{align}
The second line above is integrable in the Lebesgue measure, so the differentiation can be passed under the expectation in \eqref{eq:keybound}. The crucial remark is now that 
\begin{equation}
\mathbb E[U_{x,\tau}]=\frac{\alpha x }{\sqrt{1+\alpha^2}} +\frac{1}{\sqrt{1+\alpha^2}}\frac{\phi(\sqrt{1+\alpha^2}\tau+\alpha x)}{\Phi(\sqrt{1+\alpha^2}\tau+\alpha x)}
\end{equation}
so that integrating explicitly and making use of the the first line of \eqref{eq:partialxtrunc} produces
\begin{align}
 \partial_x \mathbb E[m(U_{x, \tau})]&=\int_{-\tau}^\infty  \partial_x\phi^x(u)m(u) du \nonumber \\ &=\alpha \sqrt{1+\alpha^2} \int_{-\tau}^\infty\phi^x(u)u\, m(u) du-  \alpha \sqrt{1+\alpha^2}\E[U_{x,\tau}]\int_{-\tau}^\infty\phi^x(u) m(u) du \nonumber \\ &=\alpha \sqrt{1+\alpha^2}
 \mbox{Cov}\left(U_{x,\tau}, m(U_{x,\tau})\right)
\end{align}
for all $|x| \leq K.$ Using the Cauchy-Schwarz inequality one can bound this covariance uniformly in $x$ and $\tau$. Indeed, the Gaussian distribution is log-concave and hence by  \cite{mailhot1988some}, Corollary 4, Var$(U_{x,\tau})$ is increasing in $\tau$; moreover since $U_{x,\tau} \stackrel{d}{\rightarrow} Z \sim N(0, (1+\alpha^2)^{-1})$ we have  Var$(U_{x,\tau})<1/(1+\alpha^2)$ for all $|x| \leq K, \tau >0$. Also obviously Var$(m(U_{x, \tau})) <\mathbb E[m(U_{x,\tau})^2]<m(-\tau)^2$ since $m$ is monotone decreasing. 
Putting this together and applying the Cauchy-Schwarz inequality leads to
\begin{align}
    & |   \partial_x R(x, \alpha, \tau) |= \frac {\phi(\tau)}{\Phi(\tau)}\Big|\partial_x\mathbb E[m(U_{x,\tau})]\Big|=|\alpha| \sqrt{1+\alpha^2}\frac{\phi(\tau)}{\Phi(\tau)}\big| \mbox{Cov}\left(U_{x,\tau}, m(U_{x,\tau})\right)\Big| \nonumber \\ &<|\alpha| \sqrt{1+\alpha^2}\frac{\phi(\tau)}{\Phi(\tau)}\sqrt{\mbox{Var}(U_{x,\tau})\mbox{Var}\left(m(U_{x,\tau})\right) }<|\alpha|
     \end{align}
for all $\alpha, \tau \in \mathbb R$ and $|x|<K$. Since the right hand constant does not depend on $K$,  the claim is proved.
\end{proof}
We have now sufficient analytical evidence to conclude strong well-posedness of the SDE \eqref{eq:SDE} with skew-normal marginals. Boundedness, uniform ellipticity and sufficiently slow rate of explosion of the derivative puts the existence and uniqueness  problem in a rather standard setting. Use of the now-classic Ambrosio-Figalli-Trevisan \emph{superposition principle}, is seemingly key to make sure that the prices of LV model match those from the skew-normal call option formula.

\begin{thm}
 The SDE \eqref{eq:SDE} with $v_L$ given by \eqref{eq:vSN} admits a unique strong martingale  solution, which is also such that for all $t, K, S_0 >0$ it holds $\E[(S_t-K)^+]=C(t,K, S_0)$, with $C$ as in \eqref{eq:call_m0}.  We denominate the corresponding process $S=(S_t)_{t \geq 0}$  the skew-normal local volatility  {(\upshape SNLV)} option pricing model.
\end{thm}

\begin{proof}

Consider the SDE  for $X_t=\log(S_t/S_0)$, which, by It\^o's Lemma is given by
\begin{equation}\label{eq:logsde}
    d X_t= -v_L(t, X_t ) dt + \sqrt{2 v_L(t, X_t)} d W_t, \qquad X_0=0. 
\end{equation}
Strong-well posedness of \eqref{eq:SDE}--\eqref{eq:vSN} is equivalent to that of 
\eqref{eq:logsde}.

Based on  \eqref{eq:vlbounds}, the coefficients of the SDE \eqref{eq:logsde} are bounded form above and away from zero, so 
 \eqref{eq:logsde}  possesses weak  solutions unique in law  by \cite{revuz2013continuous},  Corollary IX.1.14. 
Now  as a direct consequence of  Proposition \ref{lem:llipschitz}, for all $x,y \in \mathbb R$ and $t>0$ fixed, it follows the usual Lipschitz estimate \begin{equation}\label{eq:loclip}
   \left|\sqrt{2 v_L(t,x)}-\sqrt{2 v_L(t,y)}\right| ^2  < 2|v_L(t,x)-v_L(t,y)|  < 2 \sup_{\kappa \in \mathbb R} |\partial_\kappa v_L(t,\kappa)| \,|x-y| < \frac{ C}{\sqrt{t}}|x-y| 
\end{equation}
 with $C=\frac{\sigma \delta^2}{\sqrt{1-\delta^2}}$. 
 We can now apply \cite{le2006applications}, Theorem 1.3, in its time-inhomogeneous extension of Remark (c), page 22, since the necessary conditions (C)  are met using in the equations therein the following specifications: $\sigma(t,x)=\sqrt{2 v_L(t,x)},\rho(x)=x, a(x)=0$, $c(t)=C/\sqrt{t}$  and any $\delta >0$. This establishes pathwise uniqueness for any weak solution $(X,W)$ of \eqref{eq:logsde}.  Strong existence and uniqueness for the SDE \eqref{eq:logsde} then follow by the Yamada-Watanabe Theorem (\cite{revuz2013continuous}, Theorem IX.1.7).

Regarding the second statement, since $p(t,\kappa)$ in \eqref{eq:Xtpdf}, solves the Fokker-Planck equation \eqref{eq:pFP}. Now, as established $v_L$ is bounded, and therefore by \cite{figalli2008existence} Theorem 2.6, $p(t,\kappa) d \kappa$ is the law of $X_t$ under a probability measure $\nu$ which is the solution  to the martingale problem  associated with the PDE \eqref{eq:pFP}. Now by \cite{shreve1991brownian}, Proposition 5.4.6, there exists a   weak solution to \eqref{eq:SDE}--\eqref{eq:vSN}, say $\tilde X$,  on a possibly enlarged probability space $(\Omega', \nu')$ 
But then, since the solution \eqref{eq:SDE}--\eqref{eq:vSN} is strongly unique, it is also unique in law by the Yamada-Watanabe Theorem, and therefore \begin{equation}\mathcal L(X_t)=\mathcal L(\tilde X_t)=p(t,\kappa)d \kappa \end{equation} for all $t$. Thus ultimately $\E[(S_t-K)^+]=\E[(S_0e^{X_t}-K)^+]=C(t,K,S_0)$ follows.

That $S$ is a martingale and not a strict local martingale, follows from the fact that it is a supermatingale with constant expectation, the former property holding because $S$ is a  local martingale bounded from below.
\end{proof}

The local volatility surface and the $t$ sections are illustrated in Figure \ref{fig:localvol}.  We have obtained a quite simple SDE local-volatility formulation for a risk-neutral underlying $S$ with known marginals, of (geometric)  skew-normal type. In a similar vein, starting from the marginal laws of a jump-type process \cite{torricelli2025parametric}, determine a generalized-beta local volatility model, and show its strong well-posedness. An advantage of the idenitifiability  of the marginal distributions is  that the full implied volatility surface asymptotics are at hand, since those depend only on traded option prices value, and thus, ultimately, on the explicitly-known marginals of $S$. 

\begin{figure}[t]
\centering
\includegraphics[width=0.48\linewidth]{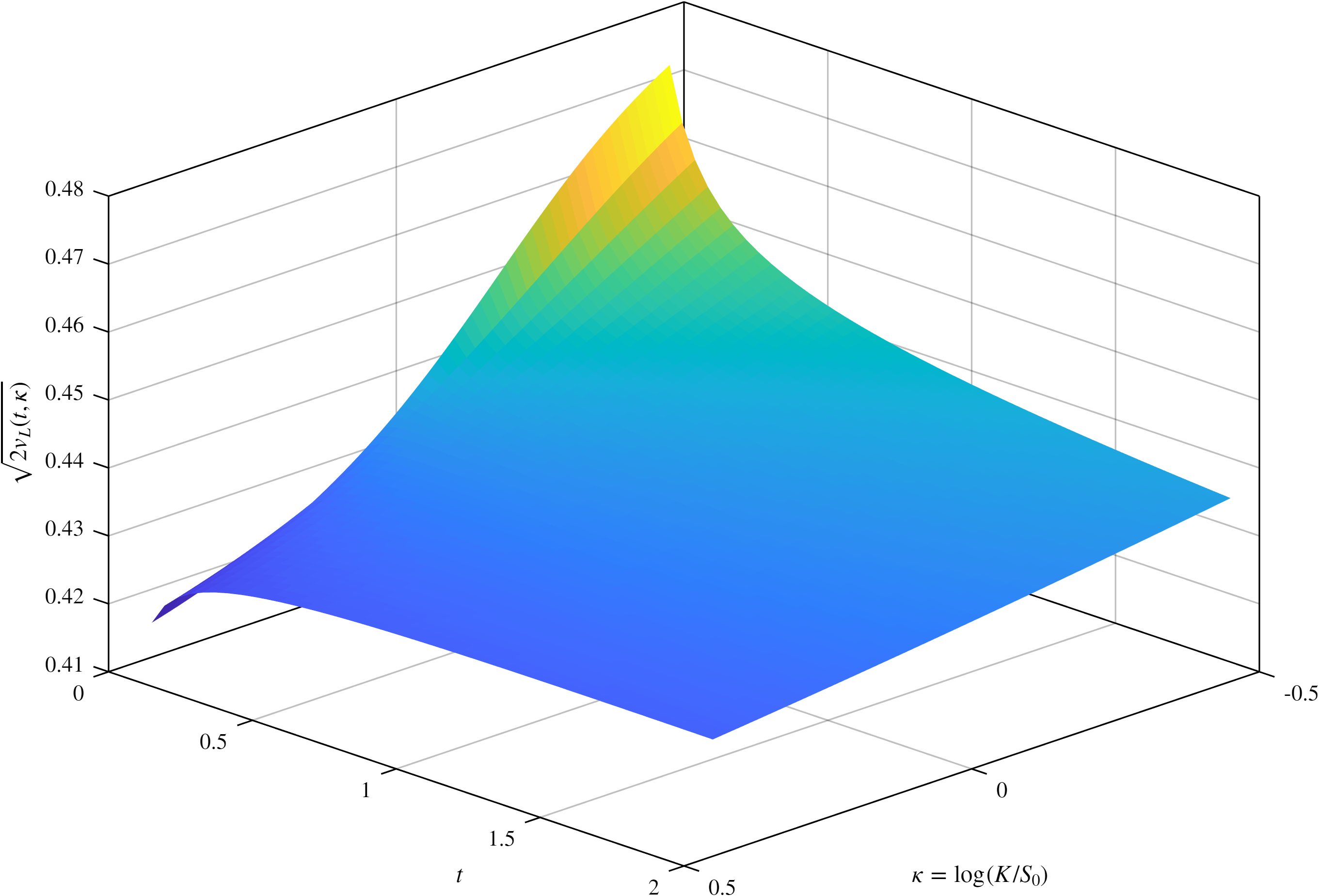}\hfill
\includegraphics[width=0.48\linewidth]{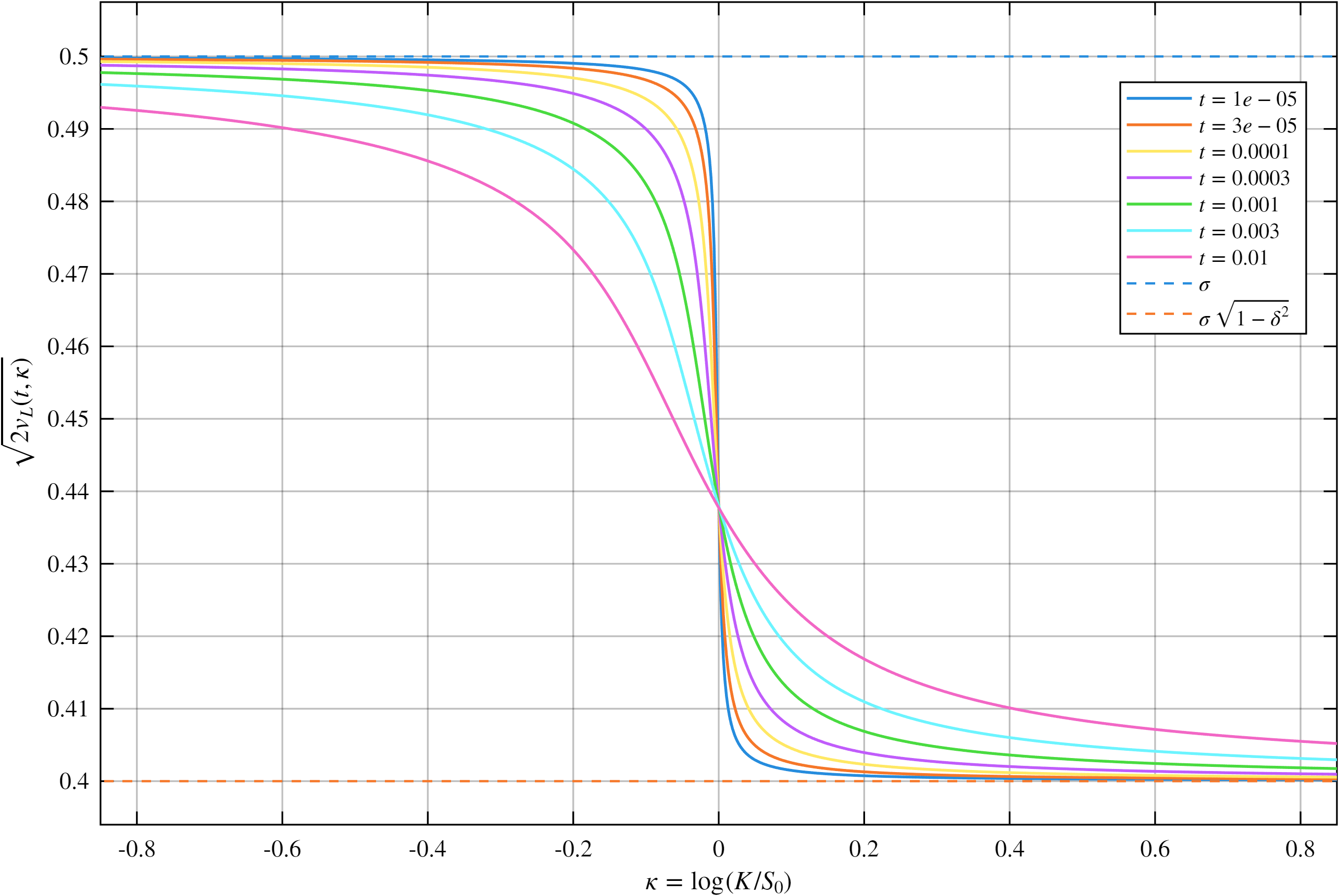}

\caption{\footnotesize
SNLV model local volatility surface. $S_0=100$, $\sigma=0.5$, $\delta=-0.6$, $r=0$.
 Left panel: local volatility surface.
Right panel: local volatility maturity slices.
}
\label{fig:localvol}
\end{figure}

\section{Implied  volatility}\label{sec:implvol}

An important element of analysis is to assess the ability of the SLNV model to reproduce the empirically observed implied volatility surface. For equity there is still general consensus --although some contrarian voices have recently arisen (\cite{guyon2022does})--  that the two prevailing features are a finite implied ATM level as time-to-maturity approaches zero, and a divergent ATM implied volatility skew  with approximately power law order of exponent approximately $-1/2$. In the following we briefly analyze where the SNLV model stands with respect to these two properties. In order to simplify the formulae, but by no means reducing generalities, we assume $r=0$.

The Black-Scholes (BS) call option price  with strike $K>0$, maturity $T>0$ and normal diffusion $\sigma>0$ is denoted by $C_{BS}(T,K,S_0, \sigma)$. The BS \emph{implied volatility} is then defined as the unique solution $\Sigma_{BS}(T,K;S_0)$ of
\begin{equation}\label{eq:impliedvol}C(T,K;S_0)=C_{BS}(T,K,S_0,\Sigma_{BS}(T,K,S_0)).
\end{equation}
Letting $\kappa=\log(K/S_0)$, when $C$ is given by \eqref{eq:call_m0} we observe the homogeneity properties
\begin{equation}\label{eq:ivolhom0}
 C(T,K,S_0)=S_0C(T,e^{\kappa},1),  \quad C_{BS}(T,K,S_0,\sigma_{BS})=S_0C_{BS}(T,e^{\kappa},1, \sigma_{BS})
\end{equation} and we have the following log-moneyness expression 
 \begin{equation}\label{eq:sdelta}
     \sigma_{BS}(T, \kappa):=\Sigma_{BS}(T, e^{\kappa},1).
 \end{equation}
We introduce further the implied volatility skew in log-moneyness as
 \begin{equation}\label{eq: classic classic skew}
     \mathcal S(T, \kappa)=\frac{\partial \sigma_{BS}}{\partial \kappa}(T, \kappa)
 \end{equation}
The implied volatility skew is a measure of the riskiness of change in spot, as maturity approaches, of the option price. The higher the slope, the faster the possible  OTM moves of the maturing option.

In the following proposition we provide the small-maturity limits of the ATM implied volatility and its skew in the SNLV model.



\begin{prop}\label{prop:skew}
Assume a zero risk-free rate is  paid on a market supporting the SNVL model. 
    Then, as $T \rightarrow 0$, we have that
    \begin{align}
     \sigma_{BS}(T,0) &\rightarrow 
        \sigma\sqrt{2\pi}
        \left[
            \phi\left(\delta \sqrt{\frac{2}{\pi}};\frac{\delta}{\sqrt{1-\delta^2}}\right)
            -
        \delta\sqrt{\frac{2}{\pi}}\Delta\left(\delta \sqrt{\frac{2}{\pi}},\frac{\delta}{\sqrt{1-\delta^2}},0\right)
        \right], \label{eq:ATMlvl} \\ 
       \mathcal S(T, 0)& \sim \sqrt{\frac{2 \pi}{T}}\left[-\frac{1}{2}+\Phi\left(\delta \sqrt{\frac{2}{\pi}}; \frac{\delta}{\sqrt{1-\delta^2}} \right)\right]. 
    \end{align}
    \end{prop}
    \begin{proof}
By e.g. \citealt{roper2009relationship}, Theorem 5.1, it holds that 
\begin{align}\label{eq:rr}
 \sigma_{BS}(T,0) \sim \sqrt{\frac{2\pi}{T}}C(T, 1, 1). 
 \end{align}
 Letting $s=  \sigma \sqrt{T}$ and denoting $\tilde C(s, 1, 1)$ the call value in hte variable $s$, using $\partial s/\partial T=\sigma/(2 \sqrt{T})$ from \eqref{eq:CpartialT2} with $r=0$  and using the limit \eqref{eq:logtaylor} we have  that
\begin{equation}
\tilde C_s(s, 1,1)=
    \phi\left(\delta \sqrt{\frac{2}{\pi}};\alpha\right)
        -
    \delta\sqrt{\frac{2}{\pi}}\Delta\left(\delta \sqrt{\frac{2}{\pi}},\alpha,0\right).
   \end{equation}
Since unlike $C$, $\tilde C$ no longer has a singularity in $s=0$, we can develop  $\tilde C(s,1,1) $ in its Taylor series around that point, and then substitute $s=\sigma \sqrt{T}$, arriving to 
\begin{equation}\label{eq:CallATMasymp}
    C(T,1,1)
    =\sigma \left[
    \phi\left(\delta \sqrt{\frac{2}{\pi}};\alpha\right)
        -
    \delta\sqrt{\frac{2}{\pi}}\Delta\left(\delta \sqrt{\frac{2}{\pi}},\alpha,0\right)
   \right]  \sqrt{T}
    +O( T), \qquad T \rightarrow 0.
\end{equation}
Using \eqref{eq:CallATMasymp} in  \eqref{eq:rr} proves \eqref{eq:ATMlvl}.
    
To show the limiting ATM skew, we recall that by \citealt{gerhold2016small}, Lemma 2.2,  the following financial expression for the ATM limiting skew holds
\begin{equation}\label{eq:gerhold}
			\mathcal S (T,0)=\sqrt{\frac{2 \pi}{T}}(1/2-DC(T,0))-\frac{\sigma_{BS}(T,0)}{2} + O(\sigma_{BS}(T,0)^2\sqrt{T} )
		\end{equation}
where 
$DC(T,\kappa):=\E[1_{\{\log S_T>\kappa\}}]=1-\Phi(b_-;\alpha)$ is      the price of the digital call in log-moneyness. 
Let further 
$DC^*(T,\kappa):=\E[S_T1_{\{\log S_T>\kappa\}}]= 1-\Phi(b_+;\alpha, \tau)$
be the value of an asset-or-nothing call. Since $C(T,1,1)=DC^*(T,0)-DC(T,0)$, substituting  \eqref{eq:rr} in \eqref{eq:gerhold} leads to 
\begin{equation}\label{eq:gerhold2}
			\mathcal S (T,0)=\sqrt{\frac{2 \pi}{T}}\left(1/2- \frac{DC^*(T,0) +DC(T,0)}{2} \right) + O(1 ).
		\end{equation}
Making use of 
\eqref{eq:logtaylor} we have that 
\begin{equation}
    \lim_{T \rightarrow 0}DC(T,0)=
    \lim_{T \rightarrow 0}DC^*(T,0)=1-\Phi\left(\delta \sqrt{\frac{2}{\pi}}; \alpha \right).
\end{equation}
 Using the above in \eqref{eq:gerhold2} shows the second limit.
\end{proof}

We gave only the leading order for $\mathcal S$, but the higher orders can be easily computed since the whole Taylor series of $DC$ and $DC^*$ can be provided. Note that from \eqref{prop:skew} $\mathcal S$ works as expected in relation to the sign of $\delta$. If $\delta <0$ the probability level of the CDF is negative and the law is negatively skewed, effects that compound into the inequality $\Phi\left(\delta \sqrt{\frac{2}{\pi}}, \alpha \right)<1/2$, so that $\mathcal S(T,0) <0$, as typically observed in equity markets. Conversely, if $\delta>0$ the SN limit law is positively skewed and the CDF taken at a positive argument, so that certainly $\mathcal S(T,0) >0$. If $\delta=0$ we correctly recover the Black-Scholes flat skew, after observing that the subleading order in \eqref{eq:gerhold} also vanishes as $T \rightarrow 0$.

Also of interest are the large strike asymptotics of the implied volatility surface. In the SNLV case those provide an interesting, non trivial example, of sublinear growth in the surface wings. Notice that since the SNLV moment generating function \eqref{eq:snmgf} exists at all $\theta$,  the SNLV model misses a maximum exponential moment, which is required for the celebrated \cite{lee2004moment} moment formula to provide meaningful wing asymptotics/limsups.
The SNLV model is thus  a concrete case in which one needs the full scope of the asymptotic analysis \cite{benaim2009regular}, based on the regular variation properties of the returns PDF, to obtain large strike limits. 

\begin{prop}
In the {\upshape SNLV} model,  as $|\kappa| \rightarrow \infty$ it holds that \begin{equation}\label{eq:lstirkes}
\sigma_{BS}(T,\kappa)
\rightarrow
\begin{cases}
\displaystyle{\sigma \sqrt{1-\delta^2}},& \mbox{if }\kappa\delta<0,\\
\sigma,& \mbox{if }\kappa\delta>0,
\end{cases}
\end{equation}
for all $T > 0$.
\end{prop}
\begin{proof}
Recall that $g:\mathbb R \rightarrow \mathbb R$ is regularly varying at $\pm \infty$ if there exists $\gamma \in \mathbb R$ such that $\lim_{x \rightarrow \pm \infty} g(\lambda x)/g(x)=\lambda^\gamma$, for all $\lambda>0$. The number $\gamma$ is the index of the variation.  Now $p(t,x)$ in \eqref{eq:Xtpdf} is such that for all $\lambda>0$ it holds that
\begin{equation}\label{eq:plambda}
\log p(T,\lambda x)=\log\left(\sqrt{\frac{2}{\ \pi}}\Phi \left(\frac{\delta}{\sqrt{1-\delta^2}}  \left( \frac{\lambda x +\mu^*(T) }{\sigma \sqrt{T}} \right)\right) \right)-\frac{(\lambda x+\mu^*(T))^2}{2 \sigma^2 T}-\log(\sigma\sqrt{T})
.
\end{equation}

As $|x| \rightarrow \infty$, recalling the expansion of the normal CDF, we obtain
\begin{equation}\label{eq:plambdalim2}
\log\left(\sqrt{\frac{2}{\ \pi}}\Phi \left(\frac{\delta}{\sqrt{1-\delta^2}}  \left( \frac{\lambda x +\mu^*(T) }{\sigma \sqrt{T}} \right)\right) \right)  \sim \begin{cases}- \frac{\delta^2}{1-\delta^2}\frac{(\lambda x+\mu^*(T))^2}{2 \sigma^2 T} , \qquad x \delta <0 \\
\log(\sqrt{2/\pi}), \qquad x \delta >0
\end{cases}
\end{equation}
Performing the asymptotic substitution  \eqref{eq:plambdalim2} in \eqref{eq:plambda} it follows that for large $|x|$ we have
\begin{equation}\label{eq:plambdalim3}
\log p(T, \lambda x )  \sim \begin{cases} \displaystyle{- \frac{1}{1-\delta^2}\frac{(\lambda x+\mu^*(T))^2}{2 \sigma^2 T}} , \qquad x \delta <0 \\
-\displaystyle{\frac{(\lambda x+\mu^*(T))^2}{2 \sigma^2 T}}, \qquad x \delta >0.
\end{cases}
\end{equation}

In either case $\log p(T, \lambda x)/\log p(T,x) \rightarrow \lambda^2$ so $   -\log p(  T, \cdot)$ has regular variation with index $\gamma=2$ at both $\pm \infty$. We can then apply theorems 1.1 and 1.2 in \cite{benaim2009regular}, deducing a sublinear behavior of $\sigma_{BS}$ at large $\kappa$, and more precisely
\begin{equation}
   \sigma_{BS}(T, \kappa) \sim \frac{| \kappa|}{\sqrt{- 2 T \log p(T, \kappa)}}, \qquad \kappa \rightarrow \pm \infty. 
\end{equation}
Calculating the right hand side above using the leading orders in \eqref{eq:plambdalim3} with $\lambda=1$, leads to \eqref{eq:lstirkes}.
\end{proof}

What the presence of two distinct OTM/ITM asymptotic regimes underlies is the monotonicity of the implied volatility maturity sections, which in turns entails that in the SNLV model smile convexity is very reduced. This means that, conceivably, financial applications of the model should be confined to equity, an asset class where typically smile is a higher order effect compared to skew.
 
That the local and implied volatility surfaces retain essentially the same geometry, as it is typically the case, is revealed by the complementary asymptotics \eqref{eq:lstirkes} and \eqref{eq:lim_T0}. Although not treated, it holds also that the small maturity limits of the local volatility surface coincide with its large strike ones, whereas for the implied volatility the limiting strike behavior at positive maturity is the  same  as \eqref{eq:lim_T0}. 

In Figure \ref{fig:impliedvol} we visualize the SNLV implied volatility and its smiles at different maturities: the monotonically decreasing maturity sections can be clearly noticed. Figure \ref{fig:atmskewlevel} shows convergence of the numerical values, as $T \rightarrow 0$, of the ATM implied volatility and skew (normalized to $\sqrt{T/2 \pi}$) to the limits taken from Proposition \ref{prop:skew}.

Using recent results of \cite{azzone2025implied}, the availability of closed expressions for option prices also makes it possible to extend the study of the implied volatility asymptotics to moneynesses different from zero. This is however less standard, and we shall leave it for further research.

\begin{figure}[t]\
\centering
\begin{subfigure}{0.48\linewidth}
    \centering
    \includegraphics[width=\linewidth]{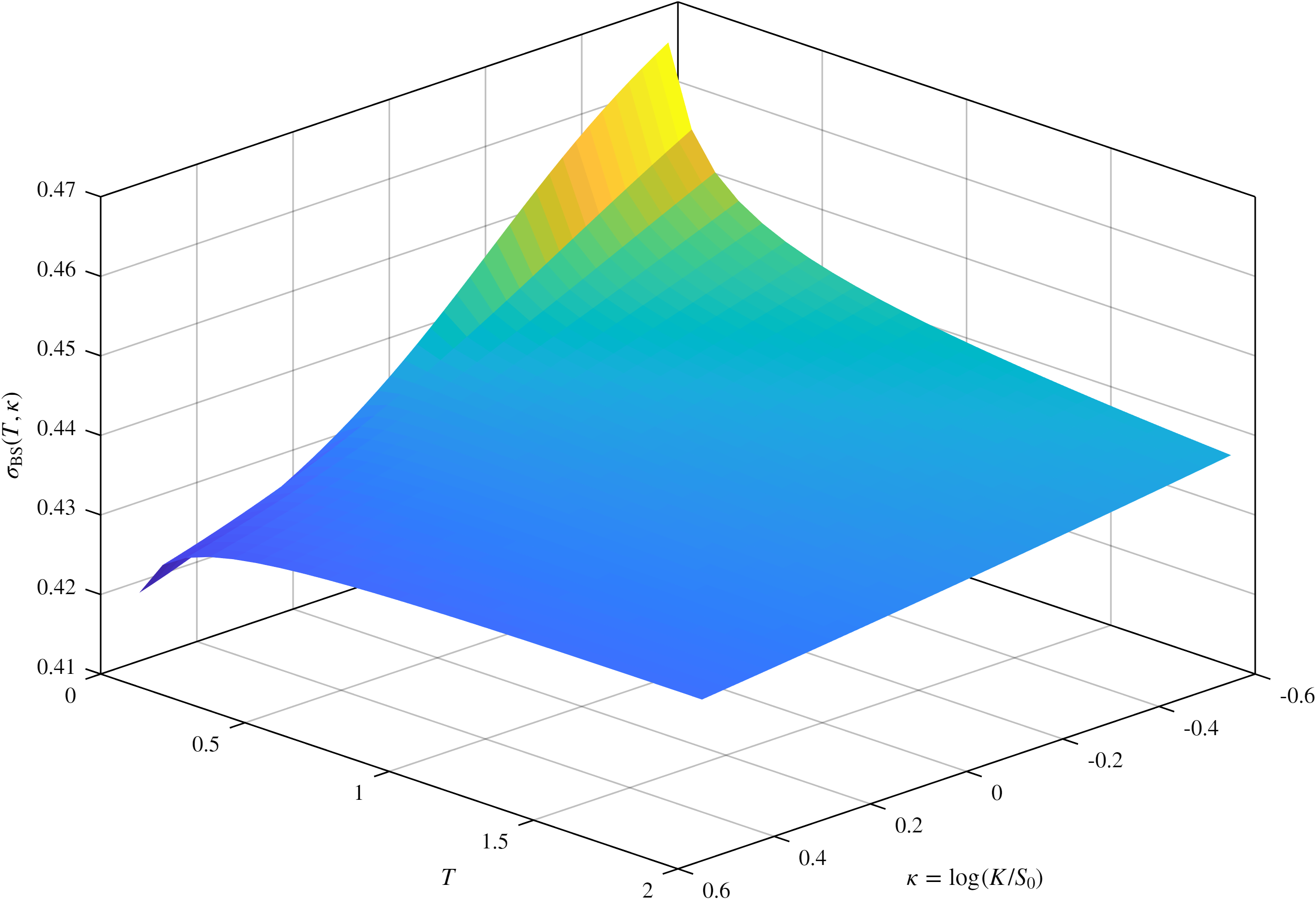}
    \label{fig:localvol-surface}
\end{subfigure}
\hfill
\begin{subfigure}{0.48\linewidth}
    \centering
    \includegraphics[width=\linewidth]{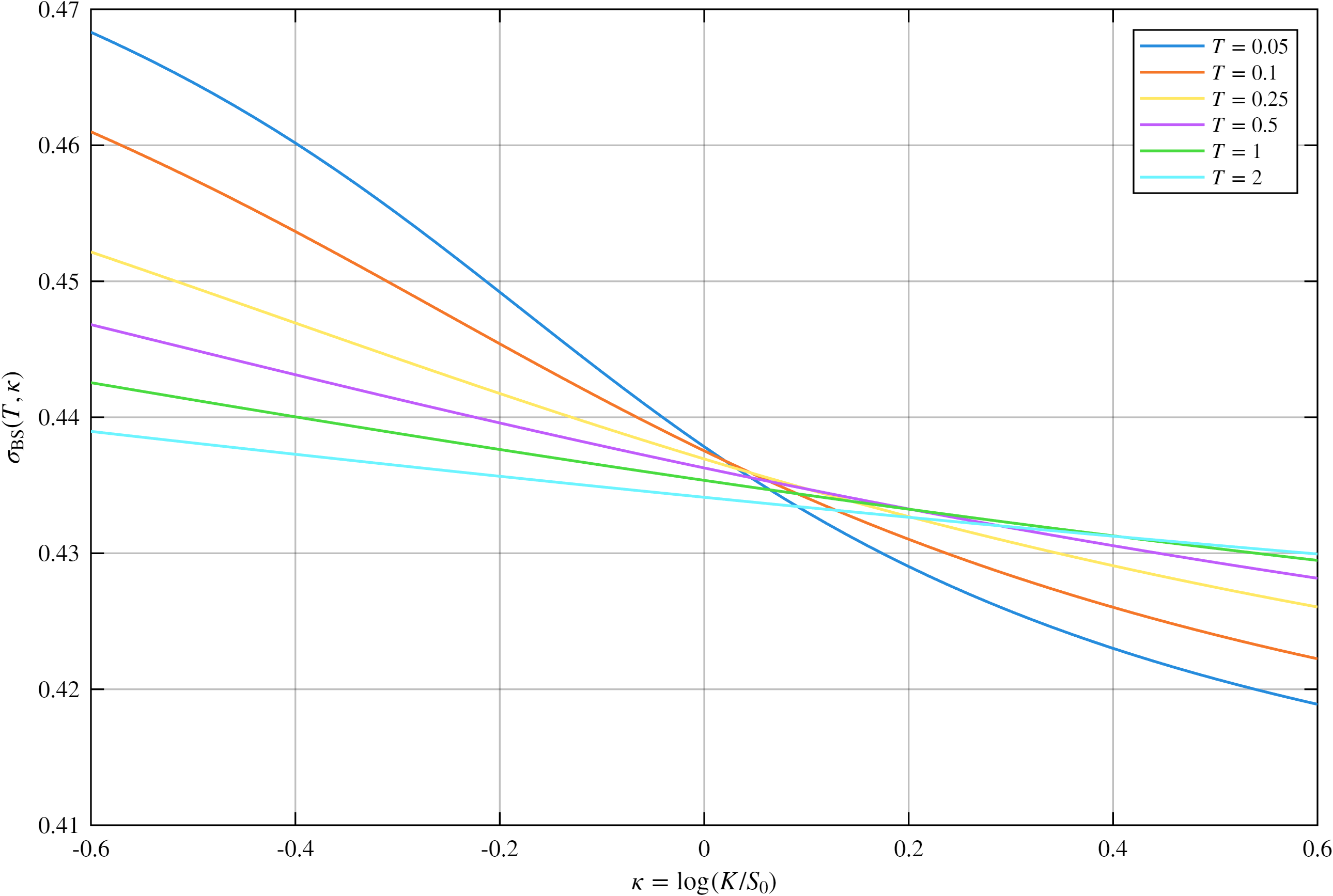}
    \label{fig:localvol-slices}
\end{subfigure}
\caption{\footnotesize
SNLV model implied volatility,
$S_0=100$, $\sigma=0.5$, $\delta=-0.6$, $r=0$. Left panel: surface. Right panel: skews.
}
\label{fig:impliedvol}
\end{figure}

\begin{figure}[htbp]
    \centering
    \includegraphics[width=0.75\textwidth]{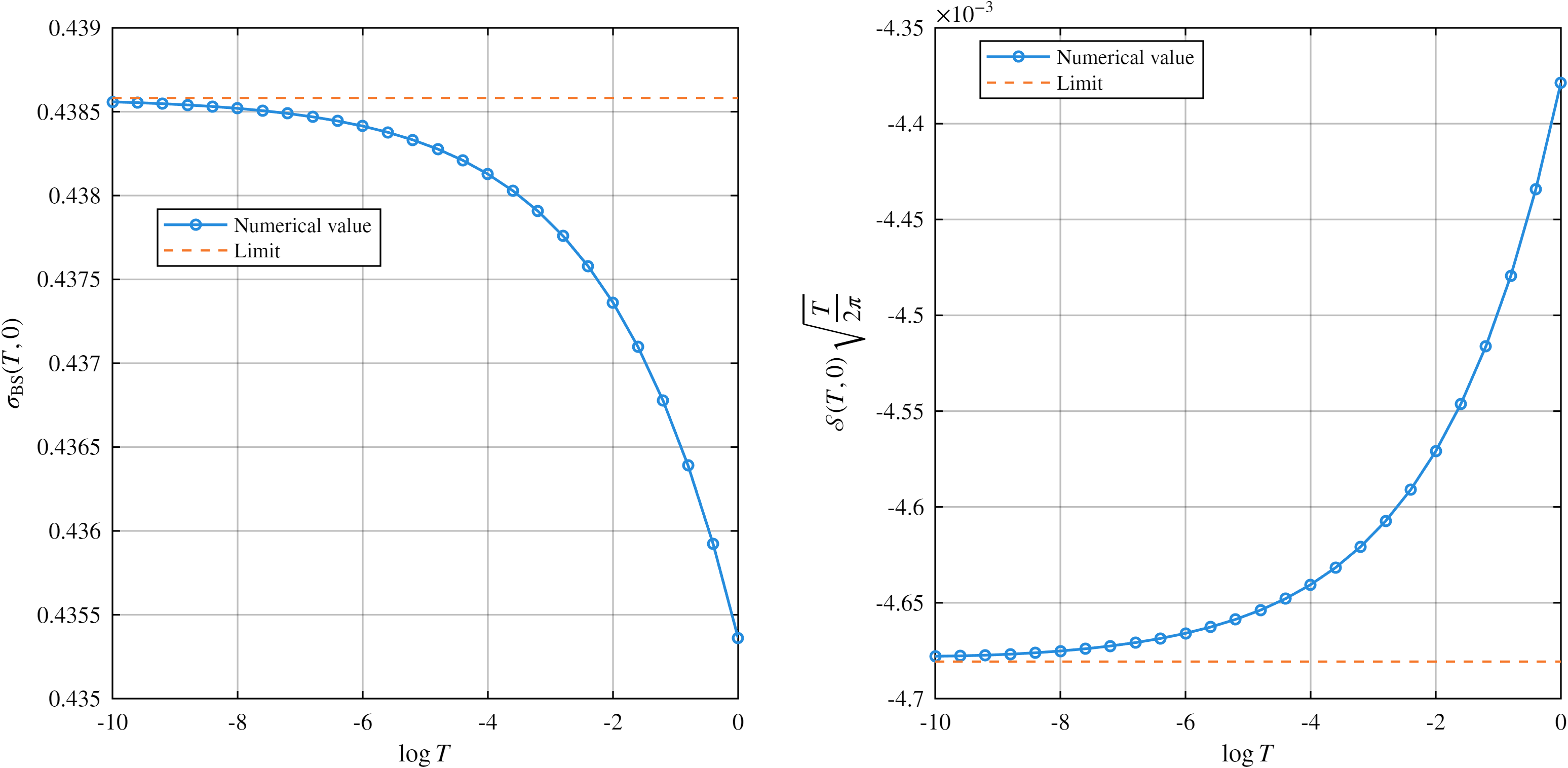}
    \caption{\footnotesize{Numerical convergence of ATM implied volatility (left panel) and skew (right panel). $S_0=100$, $\sigma=0.5$, $\delta=-0.6$, $r=0$. }}
    \label{fig:atmskewlevel}
\end{figure}

\section{Conclusion}\label{sec:conclusions}

In this work, we have reviewed and added detail to the many reasons for which a skew Brownian motion really should not be used in the returns specification of an option pricing model. We started from the clarification that the model used in most papers is not based on the process It\^o and McKean famously introduced. We then reinforced the message that skew-Brownian motion models admit arbitrages in the classic theory,  
further observing that these are among the worse possible, i.e. the profit opportunity  concealed is not only riskless, but also increasing. We  exhibited one such strategy and corrected other mistakes along the way.

As a positive contribution, we uncovered that the skew-normal marginals of the ASBM can be made to fit a no-arbitrage continuous Markovian martingale, i.e. a local volatility model, and consequently proved its strong well-posedness. The corresponding  local volatility function  is bounded and uniformly elliptic, and has a jump singularity in zero. We verified that this produces an explosion in the short term ATM implied volatility skew, consistently with what found in continuous models with discontinuous diffusion coefficients. Also, the large strike implied volatilities do not log-linearize as commonly happens, but tend to constants.

In terms of impact on the current literature, the errors  that we exposed, particularly those in \cite{zhu2018new}, jeopardize a large amount of research inspired by GSBMs. In the first place, path-dependent or simulation pricing methods  using  GSBMs  most likely need to be discarded, as a direct effect of such processes not being arbitrage free. Secondly, even if in theory the SN marginals \emph{could} form a basis to produce correct and significant European-style pricing, it is quite likely that the miscalculation error of the option price formula of \cite{zhu2018new} discussed in Appendix B did leak into such a line of work, and thus  relevant papers need to be carefully checked --and possibly amended-- case by case.

\bibliographystyle{apalike}
\bibliography{Bibliography}

\appendix

\titleformat{\section}{\normalfont\Large\bfseries}{}{0pt}{}

\section{Appendix A: some results on the skew-normal distribution}

In this appendix we gather some facts about the skew-normal law used throughout the paper. We were not able to find references for much of the content that follows, and they can arguably be considered an original contribution. The normal/skew-normal dominance Lemma \ref{lem:dom}, and the corresponding integral representation, is in our view of a certain note.

The following conjugation relation for the Gaussian PDF is clear, and is used recurringly  throughout this work
\begin{equation}\label{eq:algebralpha}
    \phi(x)\phi\left(\tau \sqrt{1+\alpha^2} + x \alpha \right) = 
    \phi(\tau)\phi\left(x \sqrt{1+\alpha^2} + \tau \alpha \right), \qquad x, \alpha,\tau \in \mathbb R. 
\end{equation}
We show next certain properties of the (extended) skew normal class relevant to us, among which some are apparently novel.

 We recall that for some $\theta \in \mathbb R$ and $X$ an absolutely continuous random  variable with PDF $f$ admitting a moment generating function at $\theta$, the exponential tilt of  $X$ by $\theta$ is the random variable $X^\theta$ with PDF $f_\theta$ given by
\begin{equation}\label{eq:tilt}
    f_\theta(x)=e^{\theta x}\frac{f(x)}{\E[e^{\theta x}]}.
\end{equation}
Tilting is strictly connected to changes of measure, and identifies the laws of dynamic process transformations such as the Girsanov and Esscher transforms, pinnacles of the mathematical theory of finance.

For the extended skew-normal family we have the following property.
\begin{prop}\label{prop:esscher}
  Let $\theta \in \mathbb \R$ and $X \sim$ {\upshape SN}$(\alpha, \tau, \mu, \sigma)$. Then $X^\theta \sim$ {\upshape SN}$\left(\alpha, \tau+\frac{\alpha \theta \sigma}{\sqrt{1+\alpha^2}},  \mu+\sigma^2 \theta, \sigma \right)$.  
\end{prop}

\begin{proof}

 The moment generating function of an SN$(\alpha, \tau)$ extended generalized skew-normal random variable $X$ exists for  all $\theta \in \mathbb R$ (see \citealt{azzalini2013skew}, Eq (2.40)) and equals to
\begin{equation}\label{eq:snmgf}
\E[e^{\theta X}]=e^{\sigma^2 \theta^2/2}\frac{\Phi\left(\tau +\frac{\alpha \sigma \theta }{\sqrt{1+\alpha^2}}\right)}{\Phi(\tau)}.
\end{equation}
Since evidently tilting commutes with location shifting it is not reductive to show the result for the SN$(\alpha, \tau, 0, \sigma)$ subclass.
Using some Gaussian PDF algebra it follows that, with $f$ an SN$(\alpha, \tau)$, in our case \eqref{eq:tilt} reads
\begin{align}
    f_\theta(x)=e^{\theta x}\frac{f\left(\frac{x}{\sigma}\right)}{\sigma \E[e^{\theta x}]}&=  e^{\theta x}\frac{\phi\left(\frac{x}{\sigma}\right)\Phi\left(\tau\sqrt{1+\alpha^2} + \alpha \frac{x}{\sigma}\right)}{{\sigma}\Phi(\tau)}
= e^{\sigma^2 \theta^2/2}\frac{\phi\left(\frac{x-\theta \sigma^2}{\sigma} \right)\Phi\left(\tau\sqrt{1+\alpha^2} + \alpha \frac{x}{\sigma}\right)}{{\sigma}\Phi(\tau)} \nonumber \\ &= e^{\sigma^2 \theta^2/2}\frac{\phi\left(\frac{x-\theta \sigma^2}{\sigma} \right)\Phi\left(\sqrt{1+\alpha^2}\left(\tau +\alpha \sigma \theta(1+\alpha^2)^{-1/2} \right) + \alpha \frac{x-\theta \sigma^2}{\sigma}\right)}{{\sigma}\Phi(\tau)}. \end{align}
Upon dividing the above by \eqref{eq:snmgf} we recognize the stated extended generalized skew normal PDF.
\end{proof}
Taking $\tau=0$ in Proposition \ref{prop:esscher}, one sees that a possible genesis of the SN$(\alpha,\tau)$ class is that of a all possible tilts of the members in SN$(\alpha)$, and hence the latter as a class is not stable with respect to tilting. 



\bigskip

We have  the following stochastic dominance lemma. The symbol $\prec_1$ stands for first-order stochastic dominance.

\begin{lem}[Normal/skew normal dominance lemma]\label{lem:dom} Let $\alpha, x,\tau \in \R.$
Then
\begin{equation}\label{eq:Delta}
    \Delta(x,\alpha,\tau):=\Phi(x \sqrt{1+\alpha^2}+ \alpha \tau)-\Phi(x;\alpha,\tau)=
\frac{\alpha}{\Phi(\tau)}
\int_{-\tau}^{+\infty}
\phi\left(\sqrt{1+\alpha^2}\,x-\alpha  u\right)\Phi(-u)\,du.
    \end{equation}
    In particular, $\mbox{\upshape{sgn}}\, \Delta(x, \alpha, \tau)= \mbox{\upshape{sgn}} (\alpha)$. Therefore if  $X \sim $ {\upshape SN}$(\alpha, \tau)$, and $Y \sim N\left(-\frac{\alpha}{\sqrt{1+\alpha^2}} \tau,\frac{1}{1+\alpha^2}\right)$ we have
$X \prec_1 Y$, if $\alpha >0$, while $Y \prec_1 X$, if $\alpha<0$.
\end{lem}

\begin{proof}
 For $\alpha=0$ there is nothing to prove as $X$ and $Y$ have N$(0,1)$ distribution. Thus, assume $\alpha \neq 0$.
and let $X \sim $ SN$(\alpha,\tau)$.
Based on \cite{azzalini2013skew}, Section 2.2.2, the following holds: for $Z,Y$ are independent standard normal random variables, then
\begin{equation}
X \stackrel{d}{=} \left\{\frac{1}{\sqrt{1+\alpha^2}}\, \left(\alpha\,Z+ Y \right) \,\biggl|\, Z>-\tau\right\}.
\end{equation}
Therefore conditioning under independence yields the alternative CDF representation
\begin{equation}\label{eq:Fintrepresentation}
\Phi(x;\alpha,\tau)
=
\mathbb E\left[
\Phi\left(\sqrt{1+\alpha^2}\,x-\alpha Z\right)
\,\middle|\, Z>-\tau
\right]=
\frac{1}{\Phi(\tau)}
\int_{-\tau}^{+\infty}
\Phi\left(\sqrt{1+\alpha^2}\,x-\alpha u\right)\phi(u)\,du.
\end{equation}
 Since by standard normal symmetry
$
\int_{-\tau}^{\infty}\phi(v)\,dv=\Phi(\tau)
$
then
\begin{align}
\Phi\left(\sqrt{1+\alpha^2}\,x+\alpha  \tau\right)-&\Phi(x;\alpha,\tau)
 \nonumber \\=&
\frac{1}{\Phi(\tau)}
\int_{-\tau}^{\infty}
\left[
\Phi\left(\sqrt{1+\alpha^2}\,x+\alpha  \tau\right)
-
\Phi\left(\sqrt{1+\alpha^2}\,x-\alpha  v\right)
\right]\phi(v)\,dv \nonumber \\ =&\frac{\alpha}{\Phi(\tau)}
\int_{-\tau}^{\infty}
\left(
\int_{-\tau}^{v}
\phi\left(\sqrt{1+\alpha^2}\,x-\alpha  u\right)\,du
\right)\phi(v)\,dv,
 \end{align}
having applied the Fundamental Theorem of Calculus in the second line. 
Since $\Phi(\tau)>0$ and  the  last integrand is positive, using Fubini's Theorem to interchange integration order leads to
\begin{equation}
\Phi\left(\sqrt{1+\alpha^2}\,x+\alpha  \tau\right)-\Phi(x;\alpha,\tau)
=
\frac{\alpha}{\Phi(\tau)}
\int_{-\tau}^{\infty}
\phi\left(\sqrt{1+\alpha^2}\,x-\alpha u\right)
\left(
\int_u^{\infty}\phi(v)\,dv
\right)\,du,
\end{equation}
and again by standard normal symmetry this leads to the desired expression \eqref{eq:Delta} for $\Delta$. The last claim is obvious by noticing that $\Delta$ is exactly the difference of the CDFs of the involved distributions.

\end{proof}
Given any arbitrary skew-normal law, it is then always possible to identify a normal distribution  stochastically dominating/being dominated by it. The difference between normal and skew-normal CDFs arises in combination with the complementary Mills ratio when differentiating the extended SN CDF with respect to the variable $\tau$.

\begin{lem}\label{lem:skewazzdertau}  For all $x, \alpha,\tau \in \mathbb R,$
we have
    \begin{align}\label{eq:derivsjewazztaqu}
    \frac{\partial}{\partial \tau} \Phi \left(x;\alpha, \tau \right)&= \Delta(x, \alpha, \tau)
    \frac{\phi(\tau)}{\Phi(\tau)}
\end{align}
with $\Delta(x, \alpha, \tau)$ as in \eqref{eq:Delta}.
\end{lem}
\begin{proof} 
Taking the derivative in $\tau$ under the integral sign is possible, so by the product rule and 
 applying  \eqref{eq:algebralpha} leads to
\begin{align}\label{eq:derivsjewazztaqu0}
    \frac{\partial}{\partial \tau} \Phi \left(x;\alpha, \tau \right)&= \Phi(\tau)^{-2}\left[\Phi(\tau) \frac{\partial}{\partial \tau} \int_{-\infty}^x\phi(y)\Phi(\tau \sqrt{1+\alpha^2}+ \alpha y)dy -\right. \nonumber \\ & \left. \phantom{xxxxxxxxxxxxxxxxx}\phi(\tau)\int_{-\infty}^x\phi(y)\Phi(\tau \sqrt{1+\alpha^2}+ \alpha y)dy \right] \nonumber \\ &=\frac{\phi(\tau)}{\Phi(\tau)}\left[  \sqrt{1+\alpha^2} \int_{-\infty}^x \phi\left( y \sqrt{1+\alpha^2}+ \alpha \tau\right) dy- \Phi(x;\alpha,\tau)  \right] \nonumber \\ &=
    \frac{\phi(\tau)}{\Phi(\tau)}\left[ 
    \Phi\left( x     \sqrt{1+\alpha^2}+ \alpha \tau\right) - \Phi(x;\alpha,\tau)  \right]
\end{align}
which is the claim.
\end{proof}

 The following Lemma reveals that the ratio of $\Delta$ to the SN PDF can be related to the expectation of the Mills ratio under a truncated normal distribution.

\begin{lem}\label{lem:boundbydelta} Let  $\Delta(x, \alpha, \tau)$ be as in \eqref{eq:Delta}, and denote  by $m(y):=\Phi(-y)/\phi(y)$, $y \in \mathbb R$ the Mills ratio. We have that
\begin{equation}\label{eq:keybound}
  \frac{\Delta(x; \alpha, \tau)} {\phi(x; \alpha, \tau)} =\frac{\alpha}{\sqrt{1+\alpha^2}}\E\left[{m(U_{x,\tau})} \right]
\end{equation}
where $U_{x,\tau}$ is a lower truncated normal random variable N$_{tr}\left(\frac{\alpha x}{\sqrt{1+\alpha^2}},\frac{1} {1+\alpha^2}  , -\tau\right)$, and moreover
\begin{equation}\label{eq:keyineq}
  \biggl |\frac{\Delta(x; \alpha, \tau)} {\phi(x; \alpha, \tau)} \biggr| \leq \frac{|\alpha|}{\sqrt{1+\alpha^2}}m(-\tau) .
\end{equation}

\end{lem}

\begin{proof}
    Using the right hand side of \eqref{eq:Fintrepresentation} and differentiating under integral sign leads to
\begin{equation}\label{Fkappa}
\phi(x;\alpha, \tau)=\partial_x \Phi(x;\alpha,\tau)
=
\frac{1}{\Phi(\tau)}
\int_{-\tau}^{\infty}
\sqrt{1+\alpha^2}\,\phi\left(\sqrt{1+\alpha^2}\,x-\alpha u\right)\phi(u)\,du.
\end{equation}
The differentiation is  justified because the integrand on the right hand side is bounded by the integrable function $\phi(u)$ times a constant. 
Substituting \eqref{Fkappa} in \eqref{eq:Delta} one obtains
\begin{equation}\label{eq:keybound2}
 \frac{\Delta(x; \alpha, \tau) }{ \phi(x; \alpha, \tau)   } = \frac{\alpha}{\sqrt{1+\alpha^2}} \,
\frac{
\int_{-\tau}^{\infty}
\phi\left(\sqrt{1+\alpha^2}\,x-\alpha  u\right)\phi(u) m(u)\,du
}{
\int_{-\tau}^{\infty}
\phi\left(\sqrt{1+\alpha^2}\,x-\alpha  u\right)\phi(u)\,du
}.
\end{equation}
But now applying \eqref{eq:algebralpha} it is 

\begin{equation}
    \phi\left(\sqrt{1+\alpha^2}\,x-\alpha u\right)\phi(u)=\phi(x) \phi\left(u \sqrt{1+\alpha^2}- \alpha x \right)=\phi(x) \phi\left(\sqrt{1+\alpha^2}\left( u -\frac{\alpha x}{\sqrt{1+\alpha^2}} \right) \right)
\end{equation}
and simplification of $\phi(x)$ appearing at both numerator and denominator of  \eqref{eq:keybound2} leads to the first equality.
Lastly, observe that
  $m$ is positive and decreasing, so that $m(u) < m(-\tau)$ for all $u \geq - \tau$ and the inequality  \eqref{eq:keyineq} follows. 

\end{proof}

\section{Appendix B: the errors in \cite{zhu2018new}}\label{sec:NoZhuHe}

As discussed in the main text, \cite{zhu2018new}
 maintain to have fixed the issues in the formula of \cite{corns2007skew}, and that an EMM for a ``drift adjusted'' version of $\SA$ in \eqref{eq:GSBM}--\eqref{eq:IMSBM} does exist. We illustrate in this appendix where the errors in such conclusion lie.  For ease of comparison with the original paper,  we strictly follow the notation in \cite{zhu2018new}.

Going through \cite{zhu2018new} they define the price process as
\begin{equation}\label{eq:GSBMcond}
S_T=S_t\exp\left( \mu (T-t) + \sigma (X_T-X_t) \right)
\end{equation}
(\cite{zhu2018new}, Equation (9)) where $X_t$ is a marginal of \eqref{eq:IMSBM}, and they attribute this formula to \cite{corns2007skew}, without specifying what the quantities $S_t$ and $X_t$ are meant to be.
If $t$ and $S_t$ are taken as fixed constants, so that  the process $S$ is implicitly defined in the interval $[t,T]$ for $T,t \geq 0$, \eqref{eq:GSBMcond} is indeed equivalent to \cite{corns2007skew}, who  write 
\begin{equation}\label{eq:GSBMcond2}
S_t=S_0\exp\left( \mu t + \sigma  X_t \right),\
\end{equation}
clarifying later that $S_0>0$.
 The heart of the problem is that  the Zhu and He paper soon changes the meaning of $t$ in \eqref{eq:GSBMcond}. Based on the initial claim of equivalence with \cite{corns2007skew}, it could only mean the initial time of the process. But in the paper $t$ quickly \emph{becomes the running time of the conditional expectation }, and starts varying in an interval $[t_0, T]$ with unspecified lower bound $t_0$ (say, $t_0=0$) and upper  bound $T$ given by the terminal asset value.

 To see when this occurs, denote with $(\mathcal F_t)_{t \geq 0}$ the filtration generated by $S$. Indeed after some formal manipulation the authors maintain that, letting $R_t=\delta W_{2,t}$, it holds
\begin{equation}\label{eq:nomart}
 \E[e^{-r (T-t) }S_T|\mathcal F_t]=S_t e^{(\mu +\sigma^2/2)(T-t)} \exp(l(t,T; R_t))
\end{equation}
with \begin{equation}
l(t,T;z)=\log\biggl[ \Phi\biggl( \frac{z+(T-t)\sigma^2\delta^2}{\sigma\delta \sqrt{T-t}}\biggr)+e^{-2z}\Phi\biggl( \frac{-z+(T-t)\sigma^2\delta^2}{\sigma\delta \sqrt{T-t}}\biggr)\biggr]
\end{equation}
(equation following (15)). The fact that a conditional expectation is present\footnote{We did remove the typo in the discount factor in the original paper, here and everywhere else.} means that time $t$ \emph{now spans an interval}, which, to begin with,  already severs the claimed equivalence between \eqref{eq:GSBMcond} and \eqref{eq:GSBMcond2}. 
If $t$  runs between 0  (or  some other $t_0 \leq T$), looked from time zero, $R_t$ and $S_t$ in \eqref{eq:nomart} are \emph{random variables}. In order for expressions such as the right hand side of \eqref{eq:nomart}, with multiple time occurrences of the same set of marginals  at various time instants, to be genuinely generated as a conditional expectation, one must \emph{show}, and not give for granted, that \eqref{eq:GSBMcond}, which in this new  understanding is just a \emph{relation} between random variables, also uniquely defines a Markovian evolution.

What the equivocation  of the meaning of $t$ -- from the beginning of the evolution window to running evaluation time--  formally, but as we shall show incorrectly, achieves, is the martingale property. By adding the nonlinear function $l(t,T; R_t)$ to the dynamics \eqref{eq:GSBMcond} (which anyway conceals further departure from the Cornell and Satchell model) after a standard Girsanov change of measure, the authors obtain 
\begin{equation}\label{eq:GSBMcondQ}
 S_T=  S_t\exp\left(  -\sigma^2 (T-t)/2 + \sigma (X_T-X_t)-l(t,T; R_t) \right)
\end{equation}
(\cite{zhu2018new}, Equation (18)). Assume that a process satisfying \eqref{eq:GSBMcondQ} for all $0 \leq t \leq T$ exists. In view of \eqref{eq:nomart} it is clear that \begin{equation}\label{eq:mart}
    \E[e^{-r(T-t)} S_T| \mathcal F_t]= S_t
\end{equation} and the martingale property is allegedly shown. However, as we will show shortly, no stochastic process can satisfy \eqref{eq:GSBMcondQ} \emph{unless $t$ is the initial time}, so that the argument just exposed cannot even be applied to begin with. We thus revert again to the usual understanding, which is equivalent to the one also laid out by \cite{corns2007skew}, of $S_t, R_t \in \mathbb R$, $l(t,T,0)$  a fixed location shift, with the process evolving on $[t,T]$.
With such understanding, since we do know the marginals of $S$, once $t$ is fixed we can easily calculate that for $t < s  < T$ we have  $\E[e^{-r(s-t)} S_s| \mathcal F_t] \neq S_t$ in law, and the martingale property  fails.

To wrap up the argument let us then show that there is no stochastic process $S=(S_t)_{t \geq 0}$ such that the expression \eqref{eq:GSBMcondQ} can hold for all times $s,t$ such that $0 \leq t < s \leq T$.  An easy way is the following. Without loss of generality we set $r=0$,
and suppose by contradiction that $S$ exists.  If that was the case we should have 
\begin{equation}\label{eq:GSBMcondtT}
 S_T= S_s\exp\left( -\sigma^2 (T-s)/2  + \sigma (X_T-X_s)-l(s,T; R_s) \right)
\end{equation}
and 
\begin{equation}\label{eq:GSBMcond0T}
 S_T= S_t\exp\left( -\sigma^2 (T-t)/2  + \sigma (X_T-X_t)-l(t,T; R_t) \right),
\end{equation}
but also
\begin{equation}\label{eq:GSBMcond0t}
S_s= S_t\exp\left( -\sigma^2 (s-t)/2  + \sigma (X_s-X_t)- l(t,s; R_t) \right).
\end{equation}
Substituting \eqref{eq:GSBMcond0t} in \eqref{eq:GSBMcondtT} we obtain
\begin{equation}\label{eq:GSBMcond0T2}
S_T= S_t\exp\left( -\sigma^2 (T-t)/2  + \sigma (X_T-X_t)-l(t,s; R_t)-l(s,T; R_s) \right).
\end{equation}
Comparing \eqref{eq:GSBMcond0T} and \eqref{eq:GSBMcond0T2} it  must hold $l(t,s; R_t)+l(s,T; R_s) = l(t,T; R_t)$, i.e. $l$ must be additive, which is clearly false. Therefore no process $S$ can satisfy \eqref{eq:GSBMcondQ}; or, better put, \eqref{eq:GSBMcondQ} can only be possibly satisfied if $t=0$, falling back to the circumstance of $S$ not being a martingale.

In addition to this main issue, there is also a subtle but profound mistake in the option pricing formula  \eqref{eq:GSBMcondQ}, that makes its applicability only partial. 
Following is the original
 \cite{zhu2018new} call option pricing formula, using the same notation as their original version presented in their Proposition 3.1, Equation (19):
\begin{equation}\label{eq:SNcallprice}
    C(S,t)=S_t M_1(b_1)-Ke^{-r(T-t)}M_2(b_2),
\end{equation}
where
$$
M_1(b_1)=\int_{b_1}^{+\infty} e^{-g(m)}L_1(u,m)\,du, \qquad M_2(b_2)=\int_{b_2}^{+\infty} L_2(u,m)\,du,
$$
$$
L_1(u,m)=\phi(u)\Phi(F_1)+e^{-2m\delta} \phi\biggl( \frac{\sigma \sqrt{T-t} \, u +2m\delta}{\sigma\sqrt{T-t}}\biggr)\Phi(F_2),
$$
$$
L_2(u,m)=\phi(u)\Phi(G_1)+ \phi\biggl( \frac{\sigma \sqrt{T-t}\,u+2m\delta}{\sigma\sqrt{T-t}}\biggr)\Phi(G_2),
$$
$$
F_1=\mbox{sgn}(\delta)\frac{\delta [\sigma \sqrt{T-t}\, u+(T-t)\sigma^2]+m }{\sigma \sqrt{(T-t)(1-\delta^2)}},$$ 
$$F_2=\text{sgn}(\delta)\frac{\delta [\sigma \sqrt{T-t}\,u+(T-t)\sigma^2+2m\delta]-m }{\sigma \sqrt{(T-t)(1-\delta^2)}},
$$
$$
G_1=\text{sign}(\delta)\frac{\delta \sigma \sqrt{T-t}\,u+m }{\sigma \sqrt{(T-t)(1-\delta^2)}}, \quad 
G_2=\text{sgn}(\delta)\frac{\delta [\sigma \sqrt{T-t}\, u+2m\delta]-m }{\sigma \sqrt{(T-t)(1-\delta^2)}},
$$
$$
b_{1,2}=\frac{\log\bigl(\frac{K}{S_t}\bigr)-\bigl(r\pm \frac{\sigma^2}{2} \bigr)(T-t)+g(m) }{\sigma\sqrt{T-t}}, \qquad m=|W_{2,t}|,
$$
\begin{equation}\label{eq:zhuhecall}    
g(m)=\log \Phi(C_1)+e^{-2m\delta}\Phi(C_2), \qquad C_{1,2}=\text{sgn}(\delta)\frac{\pm m+(T-t)\sigma^2\delta}{\sigma\sqrt{T-t}}.
\end{equation}
Firstly, we point out a clear typo: $ g(m)=\log \Phi(C_1)+e^{-2m\delta}\Phi(C_2)$ in the above should be instead $g(m)=\log \left(\Phi(C_1)+e^{-2m\delta}\Phi(C_2)\right)$, as $g$ is essentially a renaming of $l$.

As explained, equation \eqref{eq:zhuhecall} has been derived under the wrong presumption that the laws of an AGSBM started at time $t$, with initial condition $|W_{2,0}|=m>0$, serve as the time $t$ transition probabilities of some martingale process $S=(S_t)_{t \geq 0}$, for which   \begin{equation}\label{eq:CSBM}C(S_t,t)=\E[e^{-r(T-t)}(S_T-K)^+| \mathcal F_t]\end{equation}
for all $T,K>0$, $t <T$, with $C$ the function in \eqref{eq:SNcallprice}.


Even ignoring all the discussion so far on the violation of arbitrage,  there is a very simple point making apparent why $C$ in \eqref{eq:SNcallprice}--\eqref{eq:zhuhecall} cannot satisfy \eqref{eq:CSBM}.  The skew Brownian motion is not a stationary process, an argument the authors themselves correctly use to revise the original \cite{corns2007skew} call option expression. However, Zhu and He apparently do not realize stationarity \emph{is} present in \eqref{eq:SNcallprice}--\eqref{eq:zhuhecall} as such expression depends on the initial and final time is only through $T-t$. These two facts cannot come together: the conditional expectation of a measurable functional of a non-stationary process (the SBM) cannot exhibit time-stationarity. 

Now, in view of the foregoing discussion, upon assuming $t=0$,  (or $t$ any other initial fixed value), be the initial valuation time, and starting $X$ at zero when $t=0$ so that $S_0$ is the initial value, and also forcing  $m=0$, in \eqref{eq:SNcallprice}--\eqref{eq:zhuhecall} one would hope that such a formula (with the mentioned typo corrected) coincides with our pricing equation \eqref{eq:call_m0}.
Unfortunately this is not the case either. Formula \eqref{eq:SNcallprice}--\eqref{eq:zhuhecall} is the correct one only for $\delta \in [0,1]$ but it is wrong for negative $\delta$. A  way to quickly check this is to observe that the distribution classes SN$(x, \delta/\sqrt{1-\delta^2} )$ and SN$(x, -\delta/\sqrt{1-\delta^2} )$ are obviously different, whereas $C(S_t, t)$ provided in \eqref{eq:zhuhecall} \emph{only depends on} sgn$(\delta) \delta =|\delta|$. Apparently, the unfortunate  oversight $\delta^2/\delta=|\delta|$ has been committed when moving from  (23)--(24) to  (19) in Proposition 3.1. of \cite{zhu2018new}, thus making \eqref{eq:SNcallprice} strictly incorrect. Fixing the algebra with $\delta^2/\delta=\delta$ finally recovers our formula. 

The slip of adding an absolute value where there should not be one, does not of course affect the formula when $\delta$ is positive, but it is nonetheless particularly damaging, as a negative $\delta$ is what one generally looks for. For example, according to Proposition \ref{prop:skew}, in order to obtain a negatively skewed implied volatility for equity prices, in a skew-normal model we must have $\delta <0$, and the sign error in the  formula \eqref{eq:zhuhecall} would make instead the surface look identical to that corresponding to $-\delta >0$. 

In Figure \ref{fig:pcterr} we report the percentage error between Zhu and He formula and the correct one, at fixed values of $\delta<0$ and $\sigma$, across different log-moneynesses. The values are  positive and monotone increasing, so that a higher mispricing occurs for deep OTM calls. In Figure \ref{fig:pcterr2} we instead provide the ATM prices and pricing error as a function of $\delta$. Observe the incorrectly symmetric shape of the \cite{zhu2018new} in the left panel. 

\begin{figure}[htbp]
    \centering
    \includegraphics[width=0.5\textwidth]{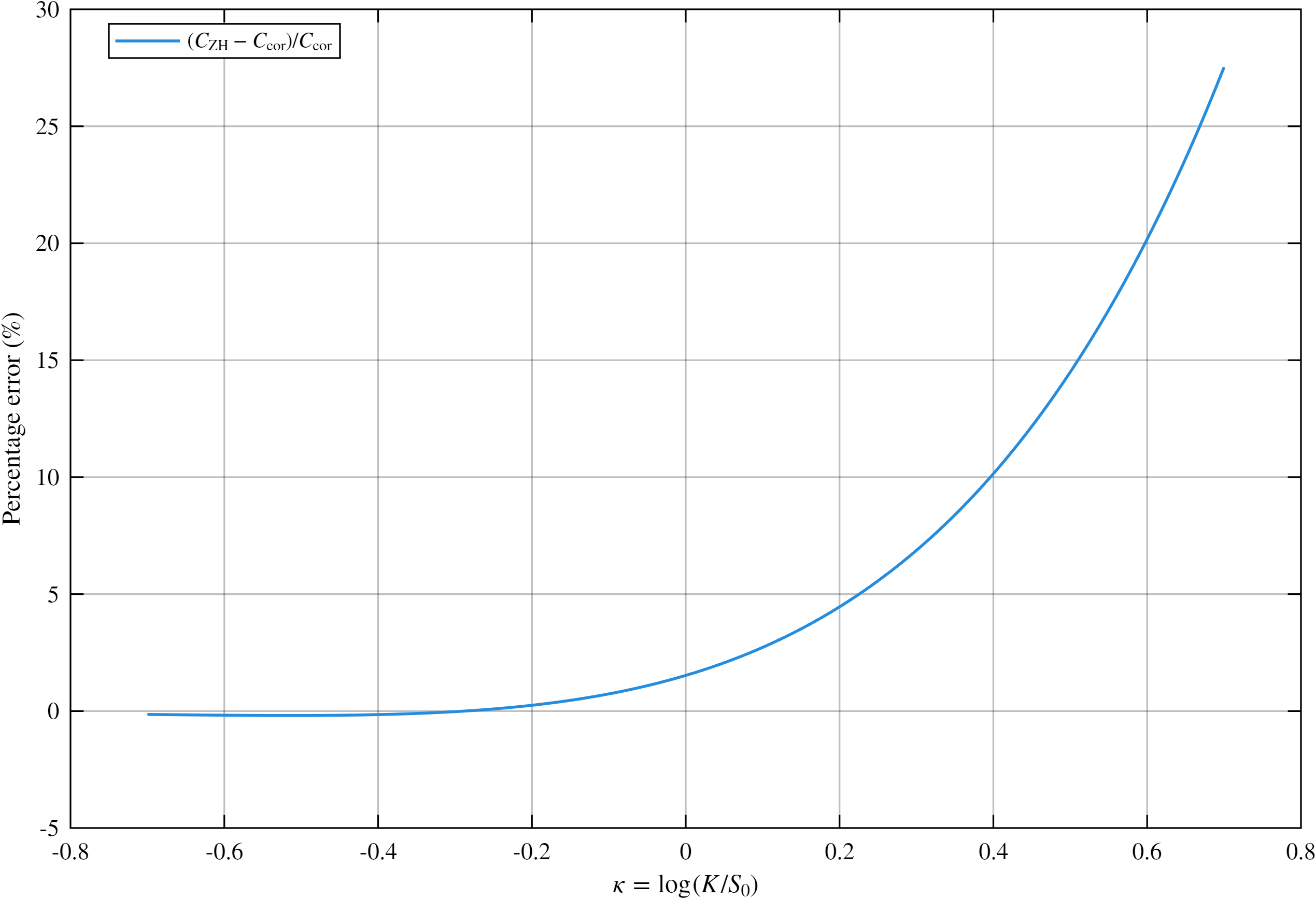}
    \caption{\footnotesize{Percentage error of option prices between Zhu and He incorrect formula and the correct  one; $\delta=-0.6$, $S_0=100$, $\sigma=0.5$, $T=1$, $r=0$.}}
    \label{fig:pcterr}
\end{figure}

\begin{figure}[htbp]
    \centering
    \includegraphics[width=0.75\textwidth]{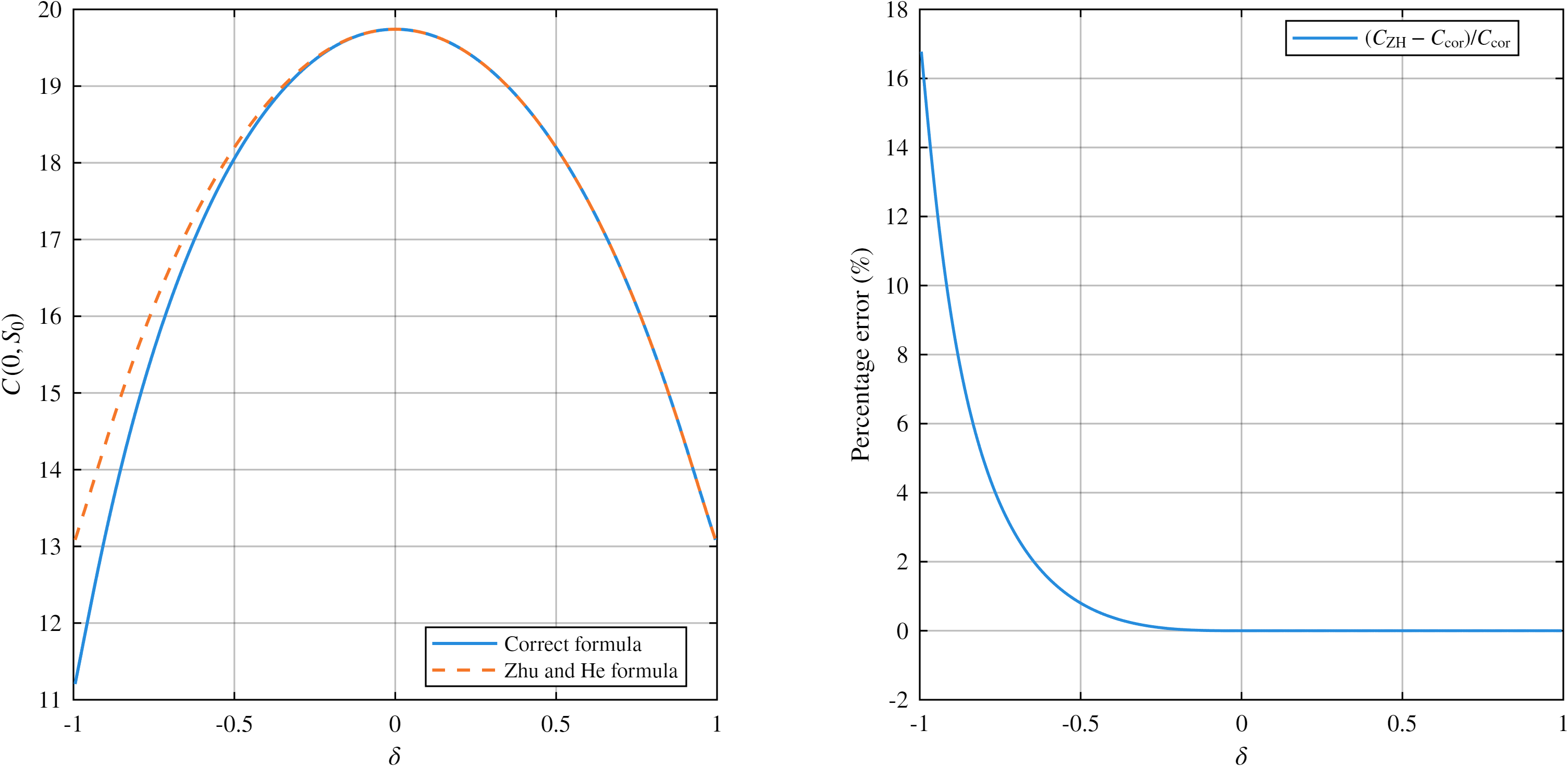
}
    \caption{\footnotesize{Comparison of ATM call option prices  as a function of $\delta$ of Zhu and He  formula and the correct one (left panel) and the corresponding percentage error (right panel). $S_0=100$, $\sigma=0.5$, $T=1$, $r=0$. }}
    \label{fig:pcterr2}
\end{figure}

\end{document}